\documentclass[letterpaper,11pt]{article}
\usepackage[utf8]{inputenc}

\usepackage[english]{babel}
\usepackage{authblk}
\usepackage[margin=1in]{geometry}

\usepackage{amsmath}
\usepackage{bbm}
\usepackage{amsthm}
\usepackage{amsfonts}
\usepackage{graphicx}
\usepackage{subcaption}
\usepackage[colorlinks=true, allcolors=blue]{hyperref}
\usepackage{xcolor}

\newtheorem{lemma}{Lemma}

\newtheorem{proposition}{Proposition}

\theoremstyle{plain}

\definecolor{DarkGreen}{rgb}{0.1,0.5,0.1}
\definecolor{DarkRed}{rgb}{0.5,0.1,0.1}
\definecolor{DarkBlue}{rgb}{0.1,0.1,0.5}
\definecolor{Gray}{rgb}{0.2,0.2,0.2}

\DeclareMathOperator*{\sigmoid}{sigmoid}
\DeclareMathOperator*{\softmax}{softmax}
\DeclareMathOperator*{\argmax}{arg\,max}
\DeclareMathOperator*{\argmin}{arg\,min}
\DeclareMathOperator*{\E}{\mathbb{E}}
\DeclareMathOperator*{\Prob}{\mathbb{P}}

\makeatletter
\newcommand\footnoteref[1]{\protected@xdef\@thefnmark{\ref{#1}}\@footnotemark}
\makeatother

\usepackage[style=trad-alpha,natbib=true,maxnames=3,maxbibnames=99]{biblatex}

\bibliography{ref.bib}

\title{\Large{Improved Bayes Risk Can Yield Reduced Social Welfare Under Competition}}

\author{Meena Jagadeesan}
\author{Michael I. Jordan}
\author{Jacob Steinhardt$^*$}
\author{Nika Haghtalab$^*$}

\affil{University of California, Berkeley}

\date{}

\begin{document}

\maketitle

\begin{abstract}
As the scale of machine learning models increases, trends such as scaling laws anticipate consistent downstream improvements in predictive accuracy. However, these trends take the perspective of a single model-provider in isolation, while in reality providers often compete with each other for users. In this work, we demonstrate that competition can fundamentally alter the behavior of these scaling trends, even causing overall predictive accuracy across users to be non-monotonic or decreasing with scale. We define a model of competition for classification tasks, and use data representations as a lens for studying the impact of increases in scale. We find many settings where improving data representation quality (as measured by Bayes risk) decreases the overall predictive accuracy across users (i.e., social welfare) for a marketplace of competing model-providers. Our examples range from closed-form formulas in simple settings to simulations with pretrained representations on CIFAR-10. 
At a conceptual level, our work suggests that favorable scaling trends for individual model-providers need not translate to downstream improvements in  social welfare in marketplaces with  multiple model providers. 
\end{abstract}

\section{Introduction}

Scaling trends in machine learning suggest that increasing the scale of a system consistently improves predictive accuracy. For example, scaling laws illustrate that increasing the number of model parameters  \citep{K20, SK20, B21} and amount of data \citep{H22} can reliably improve model performance, leading to better representations and thus better predictions for downstream tasks \citep{hernandez2021scaling}.

\def\thefootnote{*}\footnotetext{Equal  contribution}\def\thefootnote{\arabic{footnote}}

However, these scaling laws typically take the perspective of a single model-provider in isolation, when in reality, model-providers often compete with each other for users. For example, in digital marketplaces, multiple online platforms  may provide similar services (e.g., Google search vs. Bing, Spotify vs. Pandora, Apple Maps vs.\ Google) and thus compete for users  on the basis of prediction quality. A distinguishing feature of competing platforms is that users can switch between platforms and select a platform that offers them the highest predictive accuracy for their specific requests. This breaks the direct connection between predictive accuracy of a single platform in isolation and social welfare across competing platforms, and raises the question: \textit{what happens to scaling laws when model-providers compete with each other?}

\begin{figure}[t]
    \begin{subfigure}{0.49\textwidth}
        \includegraphics[scale=0.40]{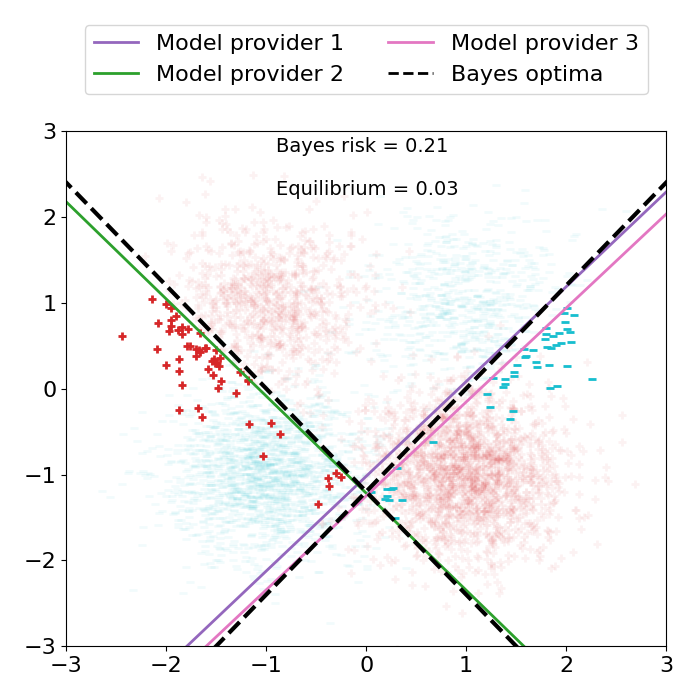}
        \caption{}
        \label{fig:theory-noiseb}
    \end{subfigure}
        \centering
    \begin{subfigure}{0.49\textwidth}
        \includegraphics[scale=0.40]{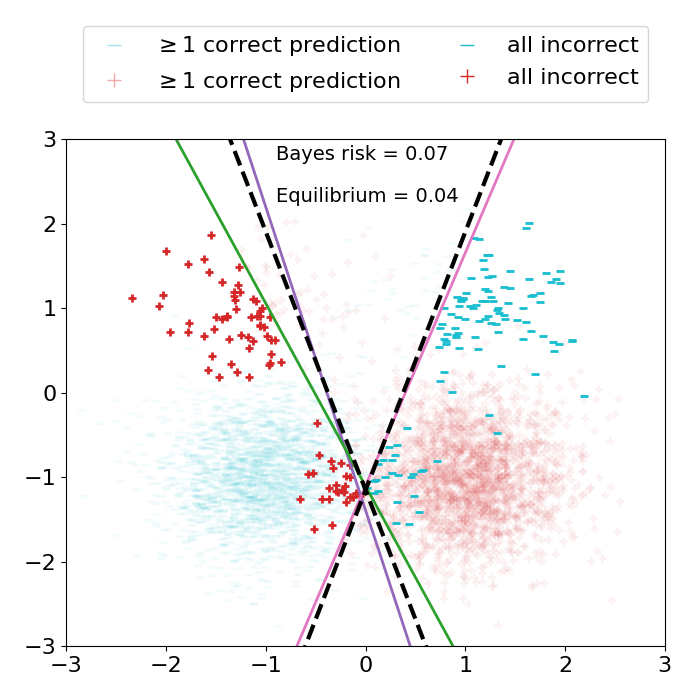}
        \caption{}
        \label{fig:theory-bayesa}
    \end{subfigure}
    \caption{Comparison of equilibrium loss on two data distributions, one with high Bayes risk (left) and one with lower Bayes risk (right). Each plot shows the linear predictors chosen at equilibrium under competition between three model-providers (solid lines), along with two approximately Bayes-optimal predictors (dashed lines). The equilibrium social loss is lower in the left plot than the right plot, even though the Bayes risk is much higher. The intuition is that approximate Bayes optima disagree on more data points in the left plot than in the right plot; thus, users have a greater likelihood of at least one predictor offering them a correct prediction, which increases the overall predictive accuracy for users (i.e., the social welfare).}

    \label{fig:intro}
\end{figure}

We show that the typical intuition about scaling laws can fundamentally break down under competition. Surprisingly, even monotonicity can be violated: increasing scale can \textit{decrease} the overall predictive accuracy (social welfare) for users. More specifically, we study increases to scale through the lens of data representations (i.e., learned features), motivated by how increasing scale generally improves representation quality \citep{BCV13}.\footnote{We are motivated by emerging marketplaces where different model-providers utilize the same pretrained model, but \textit{finetune} the model in different ways. To simplify this complex training process, we conceptualize pretraining as \textit{learning data representations (e.g., features)} and fine-tuning as \textit{learning a predictor from these representations}. In this formalization, increasing the scale of the pretrained model (e.g., by increasing the number of parameters or the amount of data) leads to improvements in data representations accessible to the model-providers during ``fine-tuning''.}
 We exhibit several multi-class classification tasks where better data representations (as measured by Bayes risk) \textit{decrease} the overall predictive accuracy (social welfare) for users, when varying data representations along several different axes.

The basic intuition for this non-monotonicity is illustrated in Figure \ref{fig:intro}. When data representations are low quality, any predictor will be incorrect on a large fraction of users, and near-optimal predictors may disagree on large subpopulations of users. Model providers are thus incentivized to choose complementary predictors that cater to different subpopulations (market segments), thus improving the overall predictive accuracy for users. In contrast, when representations are high quality, each optimal predictor is incorrect on only a small fraction of users, and near-optimal predictors likely agree with each other on most data points. As a result, model-providers are incentivized to select similar predictors, which decreases the overall predictive accuracy for users.

To study when non-monotonicity can occur, we first focus on a stylized setup that permits closed-form calculations of the social welfare at equilibrium (Section \ref{sec:theory}). Using this characterization, in three concrete binary classification setups, we show that the equilibrium social welfare can be non-monotonic in Bayes risk. In particular, we vary representations along three axes---the per-representation Bayes risks, the noise level of representations, and the dimension of the data representations---and exhibit non-monotonicity in each case (Figure \ref{fig:main}).

Going beyond the stylized setup of Section \ref{sec:theory}, in Section \ref{sec:linear} we consider linear function classes and   demonstrate empirically that the social welfare can be non-monotonic in the data representation quality. We consider binary and 10-class image classification tasks on CIFAR-10 where data representations are obtained from the last-layer representations of AlexNet, VGG16, ResNet18, ResNet34, and ResNet50,  pretrained on ImageNet. Better representations (as measured by Bayes risk) can again perform worse under competition (Figures \ref{fig:other} and \ref{fig:10class}). We  also consider synthetic data where we can vary representation quality more systematically, again finding ubiquitous non-monotonicities.

Altogether, our results demonstrate that the classical setting of a single model-provider can be a poor proxy for understanding multiple competing model-providers.  This suggest that caution is needed when inferring that increased social welfare necessarily follows from  the continuing trend towards improvements in predictive accuracy in machine learning models. Machine learning researchers and regulators should evaluate methods in environments with competing model-providers in order to reasonably assess the implications of raw performance improvements for social welfare.

\subsection{Related work}

Our work connects to research threads on the \textit{welfare implications of algorithmic decisions} and \textit{competition between data-driven platforms}.

\paragraph{Welfare implications of algorithmic decisions.} Recent work investigates \textit{algorithmic monoculture} \citep{KR21, BCKJL22}, a setting in which multiple model-providers use the same predictor. In these works, monoculture is  intrinsic to the decision-making pipeline: model-providers are given access to a shared algorithmic ranking \citep{KR21} or shared components in the training pipeline \citep{BCKJL22}. In contrast, in our work, monoculture may arise endogenously from competition, as a result of scaling trends. Model-providers are always given access to the same function classes and data, but whether or not monoculture arises depends on the quality of data representations and its impact on the incentives of model-providers. Our work thus offers a new perspective on algorithmic monoculture, suggesting that it may arise naturally in competitive settings as a side effect of improvements in data representation quality.

More broadly, researchers have identified several sources of mismatch between predictive accuracy and downstream welfare metrics. This includes \textit{narrowing} of a classifier under repeated interactions with users~\citep{HSNL18}, \textit{preference shaping} of users induced by a recommendation algorithm~\citep{CDHR22, DM22, CHRH22}, 
\textit{strategic adaptation} by users under a classifier~\citep{bruckner12pred, hardt16strat}, and the \textit{long-term impact of algorithmic decisions}~\citep{LDRSH18, LWHKBC20}. 

\paragraph{Competition between data-driven platforms.} Our work is also related to the literature on competing predictors. The model in our paper shares similarities with the work of \citet{BT17, BT19}, who studied equilibria between competing predictors. \citet{BT17, BT19} show that empirical risk minimization is not an optimal strategy for a model-provider under competition and design algorithms that compute the best-responses; in contrast, our focus is on the equilibrium social welfare and how it changes with data representation quality. The specifics of our model also slightly differ from the specifics of \citet{BT17, BT19}. In their model, each user has an accuracy target that they wish to achieve and randomly chooses between model-providers that meet that accuracy target; in contrast, in our model, each user noisily chooses the model-provider that minimizes their loss and model-providers can have asymmetric market reputations. 

Our work also relates to \textit{bias-variance games} \citep{FGHJN19} between competing model-providers. However, \citet{FGHJN19} focus on the the equilibrium strategies for the model-provider, but do not consider equilibrium social welfare for users; in contrast, our work focuses on the equilibrium social welfare. The model of \citet{FGHJN19} also differs from the model in our work. In \citet{FGHJN19}, a model-provider action is modeled as choosing an error \textit{distribution} for each user, where the randomness in the error is intended to capture randomness in the training data samples and in the predictor; moreover, the action set includes error distributions with a range of different variances. In contrast, in our population-level setup with deterministic predictors, the error distribution for every user is always a point mass (variance 0). Thus, the equilibrium characterization of \citet{FGHJN19} does not translate to our model. The specifics of the model-provider utility in the work of \citet{FGHJN19} differs slightly from our model as well. 

Other aspects studied in this research thread include competition between model-providers using \textit{out-of-box} learning algorithms that do not directly optimize for market share~\citep{GZKZ21, KGZ22, DCRMF22}, competition between model-providers selecting \textit{regularization parameters} that tune model complexity \citep{IK22}, competition between \textit{bandit algorithms} where data directly comes from users~\citep{AMSW20, JJH22}, and competition between \textit{algorithms dueling} for a user \citep{IKLMPT11}. Our work also relates to \textit{classical economic models of product differentiation} such as Hotelling's model \citep{H29, DGT79} (see \citet{ProductDifferentiation} for a textbook treatment), as well as the emerging area of \textit{platform competition}~\citep[see, e.g.,][]{J21, C21}.

\section{Model}\label{sec:model}

We focus on a multi-class classification setup with input space $X\subseteq \mathbb{R}^d$ and output space $Y = \left\{0,1, 2, \ldots, K-1\right\}$. Each user has an input $x$ and a corresponding true output $y$, drawn from a distribution $\mathcal{D}$ over $X \times Y$. Model providers choose predictors $f$ from some model family $\mathcal{F} \subseteq (\Delta(Y))^X$ where $\Delta(Y)$ is the set of distributions over $Y$. A user's loss given predictor $f$ is $\ell(f(x), y) = \mathbb{P}[y \neq f(x)]$. In Section \ref{sec:theory}, we take $\mathcal{F} = \{0,1, 2, \ldots, K-1\}^X$ to be all deterministic functions mapping inputs to classes, while in Section \ref{sec:linear} we consider linear predictors of the form $f(x) = \softmax(Wx + b)$. 

We study competition between $m \ge 2$ model-providers for users, building on the model of \citet{BT17, BT19}. We index the model-providers by $[m] := \left\{1, 2, \ldots, m\right\}$, and let $f_j$ denote the predictor chosen by model provider $j$. After the model-providers choose predictors $f_1, \ldots, f_m$, each user then chooses one of the $m$ model-providers to link to, based on prediction accuracy. Model-providers aim to optimize the number of users that they win. (We note that this model is stylized and will make several simplifying assumptions; we defer a detailed discussion of the implications of these assumptions to Section \ref{sec:modeldiscussion}.)

\paragraph{User decisions.} Users noisily pick the model-provider offering the best predictions for them. That is, a user with representation $x$ and true label $y$ chooses a model-provider $j^*(x,y)$ such that the loss $|y - f_{j^*(x,y)}(x)|$ is the smallest across all model-providers $j \in [m]$, subject to noise in user decisions. More formally, we model user noise with the logit model \citep{train2002discrete}, also known as the Boltzmann rationality model:
\begin{equation}
\label{eq:userchoice}
\Prob[j^*(x,y) = j] = \frac{e^{- \ell(f_j(x), y) / c}}{\sum_{j'=1}^m e^{- \ell(f_{j'}(x), y) / c}}, 
\end{equation}
where $c > 0$ denotes a noise parameter. We extend this model to account for uneven market reputations across decisions in Section \ref{subsec:unequal}.

\paragraph{Model provider incentives.} A model-provider's utility is captured by the market share that they win. That is, model-provider $j$'s utility is 
\[u(f_j; \mathbf{f}_{-j}) := \E_{(x,y) \sim \mathcal{D}} \left[\Prob[j^*(x,y) = j] \right],\]
where $\textbf{f}_{-j}$ denotes the predictors chosen by the other model-providers and where the expectation is over $(x,y)$ drawn from $\mathcal{D}$. Since the market shares always sum to one, this is a constant-sum game.

Each model-provider chooses a \textit{best response} to the predictors of other model-providers. That is, model-provider $j$ chooses a predictor $f^*_j$ such that 
\[f^*_j \in \argmax_{f_j \in \mathcal{F}} u(f_j; \mathbf{f}_{-j}). \] The best-response captures that model-providers optimize for market share. In practice, model-providers may do so via {A/B testing} to steer towards predictors that maximize profit, or by {actively collecting data} on market segments where competitors are performing poorly.

We study market outcomes $\mathbf{f} = (f^*_1, f^*_2, \ldots, f^*_m)$ that form a Nash equilibrium. Recall that $(f^*_1, f^*_2, \ldots, f^*_m)$ is a \textit{pure strategy Nash equilibrium} if for every $j \in [m]$, model-provider $j$'s predictor is a best-response to $\mathbf{f}^*_{-j}$: that is, $f^*_j \in \argmax_{f_j\in \mathcal{F}} u(f_j; \mathbf{f}^*_{-j})$. In well-behaved instances, pure-strategy equilibria exist (see theoretical results in Section \ref{sec:theory} and simulation results in Section \ref{sec:linear}). However, for our results in Section \ref{subsec:unequal}, we must turn to mixed strategy equilibria where model-providers instead choose distributions $\mu_j$ over $\mathcal{F}$.

\paragraph{Quality of market outcome for users.}
We are interested in studying the quality of a market outcome $\mathbf{f} = (f_1, f_2, \ldots, f_m)$ in terms of user utility. The quality of $\mathbf{f}$ is determined by the overall \textit{social loss} that it induces on the user population, after users choose between model-providers:
\begin{equation}
\label{eq:sociallossdef}
  \texttt{SL}(f_1, \ldots f_m) := \E[\ell(f_{j^*(x,y)}(x),y)]. 
\end{equation}
When $f^*_1, \ldots, f^*_m$ is a Nash equilibrium, we refer to $\texttt{SL}(f^*_1, \ldots f^*_m)$ as the \textit{equilibrium social loss.}

Our goal is to study how the equilibrium social loss changes when the representation quality (i.e., the quality of the input representations $X$) improves. We formalize representation quality as the minimum risk $\texttt{OPT}_{\text{single}}$ that a single model-provider could have achieved on the distribution $\mathcal{D}$ with the model family $\mathcal{F}$. This means that $\texttt{OPT}_{\text{single}}$ is equal to the Bayes risk:  
\[\texttt{OPT}_{\text{single}} := \min_{f \in \mathcal{F}} \E\left[ \ell(f(x),y)\right].\]
In the following sections, we show that the equilibrium social loss $\texttt{SL}(f^*_1, \ldots f^*_m)$ can be non-monotonic in the representation quality (as measured by $\texttt{OPT}_{\text{single}}$), when representations are varied along a variety of axes.

\section{Non-monotonicity of Equilibrium Social Loss in a Stylized Setup}\label{sec:theory}

To understand when non-monotonicity can occur, we first consider a stylized setup (described below) that permits closed-form calculations of the social loss. We first show a simple mathematical example that illustrates non-monotonicity (Section \ref{subsec:example}). We characterize the equilibrium social loss in this setup for binary classification (Section \ref{subsec:idealized}), and apply this characterization to three concrete setups that vary representation quality along different axes (Section \ref{subsec:examples}): we show that the equilibrium social loss can be non-monotonic in Bayes risk in all of these setups (Figures~\ref{fig:theory-noise}-\ref{fig:theory-dimension}). Finally, we extend our theoretical characterization from Section \ref{subsec:idealized} to setups with more than 2 classes (Section \ref{subsec:idealizedmulticlass}), and we extend our model and results to model-providers with unequal market reputations (Section \ref{subsec:unequal}).

\paragraph{Specification of stylized setup.} Assume the input space $X$ is finite and let $\mathcal{F} = \mathcal{F}^{\text{multi-class}}_{\text{all}}$ contain {all deterministic functions} from $X$ to $\left\{0,1, \ldots, K-1\right\}$. For simplicity, we also assume that users make noiseless decisions (i.e., we take $c \rightarrow 0$), so a user's choice of model-provider $j^*(x,y)$ is specified as follows: 
\begin{equation}
    \label{eq:userchoicespecific}
    \Prob[j^*(x,y) = j] = 
\begin{cases}
0 & \text{ if } j \not\in \argmin_{j' \in [m]} \mathbbm{1}[y \neq f_{j'}(x)] \\
\frac{1}{\left|\argmin_{j' \in [m]} \mathbbm{1}[y \neq f_{j'}(x)] \right|}
& \text{ if } j \in \argmin_{j' \in [m]} \mathbbm{1}[y \neq f_{j'}(x)].
\end{cases}
\end{equation}
In other words, users pick the model-provider with minimum loss, choosing randomly in case of ties. We show that pure strategy equilibria are guaranteed to exist in this setup. 
\begin{proposition}
\label{prop:existence}
Let $X$ be a finite set of representations, let there be $K \ge 2$ classes, let $\mathcal{F} = \mathcal{F}^{\text{multi-class}}_{\text{all}}$, and let $\mathcal{D}$ be the distribution over $(X,Y)$. Suppose that user decisions are noiseless (i.e., user decisions are given by \eqref{eq:userchoicespecific}). For any $m \ge 2$, there exists a pure strategy equilibrium. 
\end{proposition}

\subsection{Simple mathematical example of non-monotonicity}\label{subsec:example}

We show a simple example where improving data representation quality (i.e. Bayes risk) reduces the equilibrium social welfare. Consider a distribution over binary labels given by $\mathbb{P}[Y = 1] = 0.6$ and $\mathbb{P}[Y = 0] = 0.4$, and suppose that there are $m = 3$ model-providers. We consider two different sets of representations $X_1$ and $X_2$, which give rise to two different distributions $\mathcal{D}_1$ over $X_1 \times Y$ and $\mathcal{D}_2$ over $X_2 \times Y$ satisfying  $\mathbb{P}[Y = 1] = 0.6$ and $\mathbb{P}[Y = 0] = 0.4$. 

Suppose that $X_1 = \left\{x_0 \right\}$ consists of the trivial representation which provides no information about users. The distribution $\mathcal{D}_1$ is specified by $\mathbb{P}_{\mathcal{D}_1}[Y = 1 \mid X_1 = x_0] = 0.6$ and $\mathbb{P}_{\mathcal{D}_1}[Y = 0 \mid X_1 = x_0] = 0.4$. In this case, the Bayes risk is $0.4$. Moreover,  
it is not difficult to see that $f^*_1(x_0) = f^*_2(x_0) = 1$ and $f^*_3(x_0) = 0$ is an equilibrium. (The reason that $f_1(x_0) = f_2(x_0) = f_3(x_0) = 1$ is not an equilibrium is that model provider 3 would deviate to $f^*_3(x_0) = 0$ and 
\textit{increase} their utility from $1/3$ to $0.4$.) Since the model-providers collectively offer both labels for the representation $x_0$, each user has the option to choose either label, so the equilibrium social loss $\texttt{SL}(f^*_1, f^*_2, f^*_3) = 0$. 

Next, suppose that $X_2 = \left\{x_1,x_2\right\}$ consists of binary representations that provide some nontrivial information about users.
In particular, the distribution $\mathcal{D}_2$ is specified by equally likely representations $\mathbb{P}_{\mathcal{D}_2}[X_2 = x_1] = \mathbb{P}_{\mathcal{D}_2}[X_2 = x_2] = 0.5$. The conditional distribution $Y \mid X_2$ is specified by $\mathbb{P}_{\mathcal{D}_2}[Y = 1 \mid X_2 = x_1] = 0.4$, $\mathbb{P}_{\mathcal{D}_2}[Y = 0 \mid X_2 = x_1] = 0.6$, $\mathbb{P}_{\mathcal{D}_2}[Y = 1 \mid X_2 = x_2] = 0.8$, and $\mathbb{P}_{\mathcal{D}_2}[Y = 0 \mid X_2 = x_2] = 0.2$. In this case, the Bayes risk goes down to $0.3$. Moreover, it is not difficult to see that $f^*_1(x_1) = f^*_2(x_1) = 0$,  $f^*_3(x_1) = 1$, and $f^*_1(x_2) = f^*_2(x_2) = f^*_3(x_2) = 1$ is an equilibrium. (Intuitively, the reason that $f^*_1(x_2) = f^*_2(x_2) = f^*_3(x_2) = 1$ occurs at equilibrium in this setup is that no model provider $i \in [m] = \left\{1, 2, 3\right\}$ wants to deviate to $f_i(x_2) = 0$, since this would \textit{decrease} their utility on $X_2 = x_2$ from $1/3$ to $0.2$.) Since users with representation $x_2$ no longer have the option to choose the label of $0$, the equilibrium social loss is $\texttt{SL}(f^*_1, f^*_2, f^*_3) = 0.1$.

As a result, even though the Bayes risk is lower for representations in the second setup than for the representations in the first setup, the equilibrium social loss is higher. This instantation thus provides a simple mathematical example where non-monotonicity occurs. In the remaining sections, we consider more general setups that elucidate what factors drive non-monotonicity. 

\subsection{Characterization of the equilibrium social loss for binary classification}\label{subsec:idealized}

To generalize the above example, we analyze general instantations of the stylized setup, focusing first on binary classification. Let $\mathcal{F}^{\text{binary}}_{\text{all}}$ denote the function class $\mathcal{F}^{\text{multi-class}}_{\text{all}}$ in the special case of $K = 2$ classes. Since $\mathcal{F}^{\text{binary}}_{\text{all}}$ lets model-providers make independent predictions about each representation $x$, the only source of error is noise in individual data points. To capture this, we define the \emph{per-representation Bayes risk} $\alpha(x)$ to be: 
\begin{equation}
\label{eq:alphabayes}
 \alpha(x) := \min(\Prob(y = 1 \mid x), \Prob(y = 0 \mid x)). 
\end{equation}
The value $\alpha(x)$ measures how random the label $y$ is for a given representation $x$. As a result, $\alpha(x)$ is the minimum error that a model-provider can hope to achieve on the given representation $x$. Increasing $\alpha(x)$ increases the Bayes risk $\texttt{OPT}_{\text{single}}$: in particular,  $\texttt{OPT}_{\text{single}}$ is equal to the average value $\E[\alpha(x)]$ across the population. The equilibrium social loss, however, depends on other aspects of $\alpha(x)$.

We characterize the equilibrium social loss in terms of the per-representation Bayes risks in the following proposition. Our characterization focuses on  pure-strategy equilibria, which are guaranteed to exist in this setup (see Proposition \ref{prop:existence}).
\begin{proposition}
\label{prop:idealized}
Let $X$ be a finite set, let $K =2$, and let $\mathcal{F} = \mathcal{F}^{\text{binary}}_{\text{all}}$. Suppose that user decisions are noiseless (i.e., user decisions are given by \eqref{eq:userchoicespecific}). 
Suppose also that $\alpha(x) \neq 1/m$ for all $x \in X$.\footnote{When $\alpha(x) = 1/m$, there turn out to be multiple pure-strategy equilibria with different social losses.} 
At any pure strategy Nash equilibrium $f^*_1, \ldots, f^*_m$, the social loss $\texttt{SL}(f^*_1, \ldots, f^*_m)$ is equal to:
\begin{equation}
\label{eq:socialloss} 
\texttt{SL}(f^*_1, \ldots, f^*_m) =  \E_{(x,y) \sim \mathcal{D}} \left[\alpha(x) \cdot \mathbbm{1}[\alpha(x) < 1/m] \right]. 
\end{equation}
\end{proposition}
\noindent The primary driver of Proposition \ref{prop:idealized} 
is that as the per-representation Bayes risk $\alpha(x)$ decreases, the equilibrium predictions for $x$ go from \textit{heterogeneous} (different model-providers offer different predictions for $x$) to \textit{homogenous} (all model-providers offer the same prediction for $x$). In particular, if $\alpha(x)$ is below $1/m$, then all model-providers choose the Bayes optimal label $y^* = \argmax_{y'} \Prob[y' \mid x]$, so predictions are homogeneous; on the other hand, if $\alpha(x)$ is above $1/m$, then at least one model-provider will choose $1- y^*$, so predictions are heterogeneous. 
When predictions are heterogeneous, each user is offered perfect predictive accuracy by some model-provider, which results in zero social loss. On the other hand, if predictions are homogeneous and all model-providers choose the Bayes optimal label, the social loss on $x$ is the per-representation Bayes risk $\alpha(x)$. Putting this all together, the equilibrium social loss takes the value in \eqref{eq:socialloss}. We defer a proof of Proposition \ref{prop:idealized} to Appendix \ref{appendix:theory}.

\subsection{Non-monotonicity along several axes of varying representations}\label{subsec:examples}

\begin{figure}[t]
    \centering
    \begin{subfigure}[b]{0.32\textwidth}
        \includegraphics[scale=0.12]{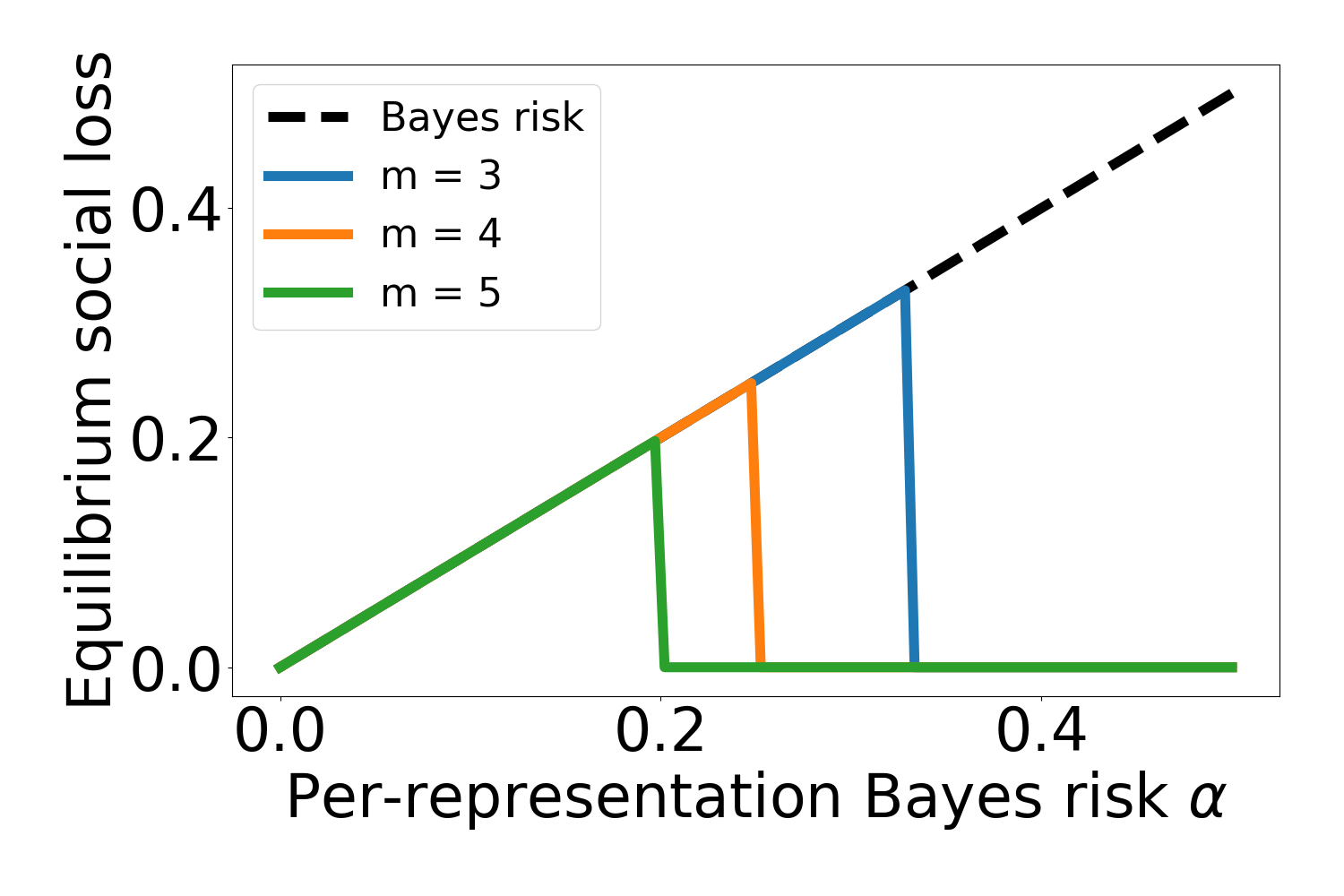}
        \caption{}
        \label{fig:theory-bayes}
    \end{subfigure}
    \begin{subfigure}[b]{0.32\textwidth}
        \includegraphics[scale=0.12]{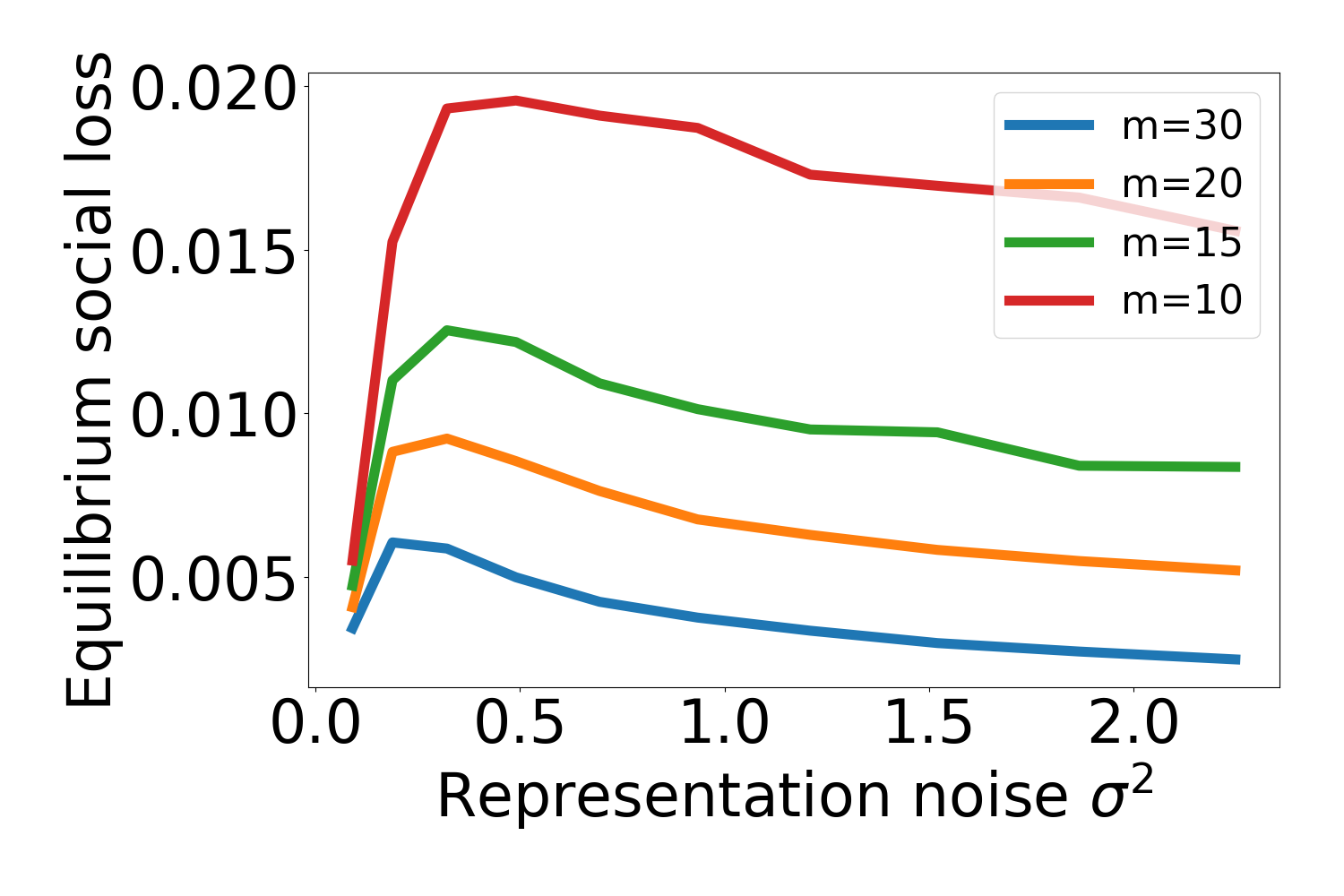}
        \caption{}
        \label{fig:theory-noise}
    \end{subfigure}
    \begin{subfigure}[b]{0.32\textwidth}
        \includegraphics[scale=0.12]{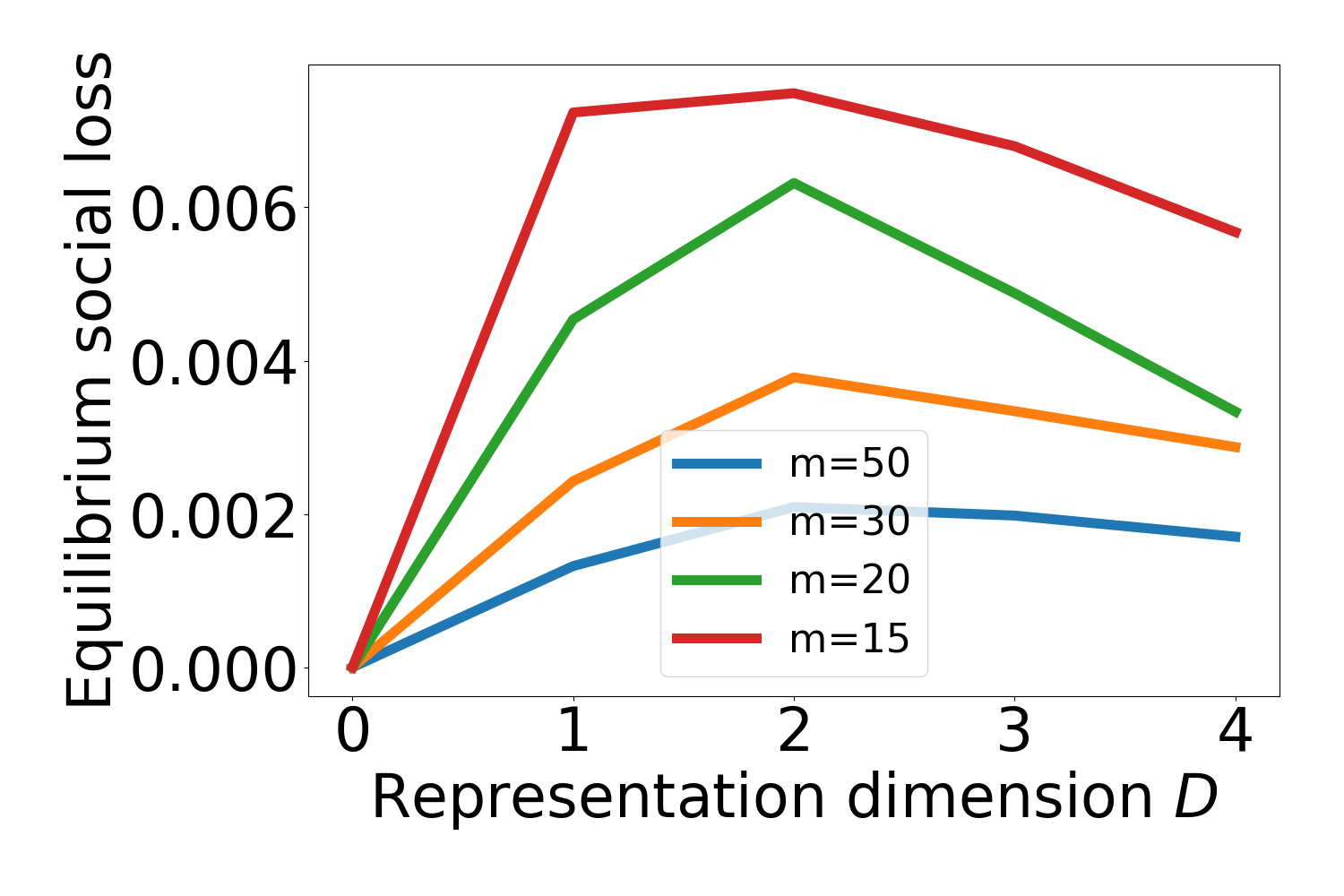}
        \caption{}
        \label{fig:theory-dimension}
    \end{subfigure}
    \begin{subfigure}[b]{0.32\textwidth}
        \includegraphics[scale=0.12]{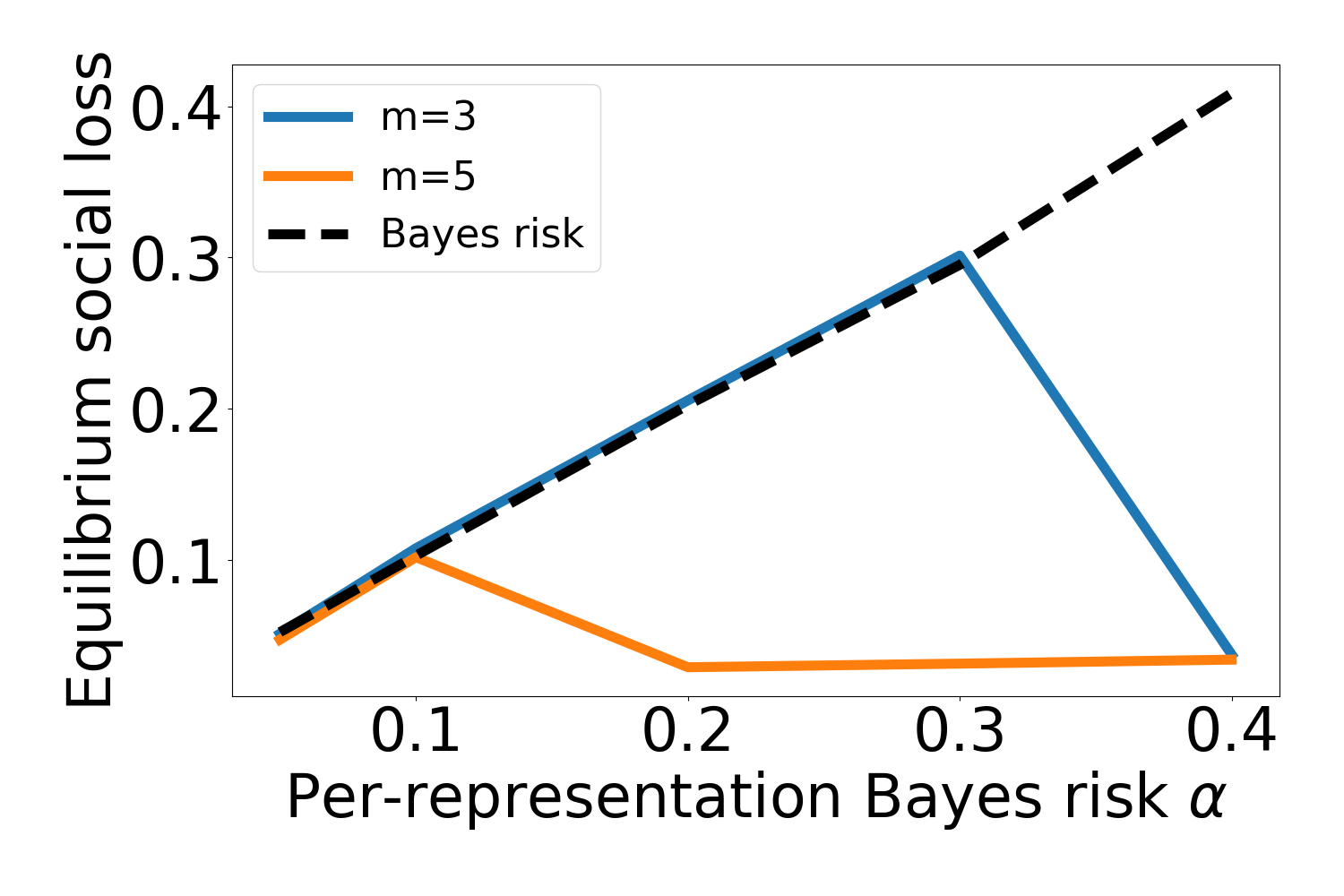}
        \caption{}
        \label{fig:simulations-bayes}
    \end{subfigure}
    \begin{subfigure}[b]{0.32\textwidth}
        \includegraphics[scale=0.12]{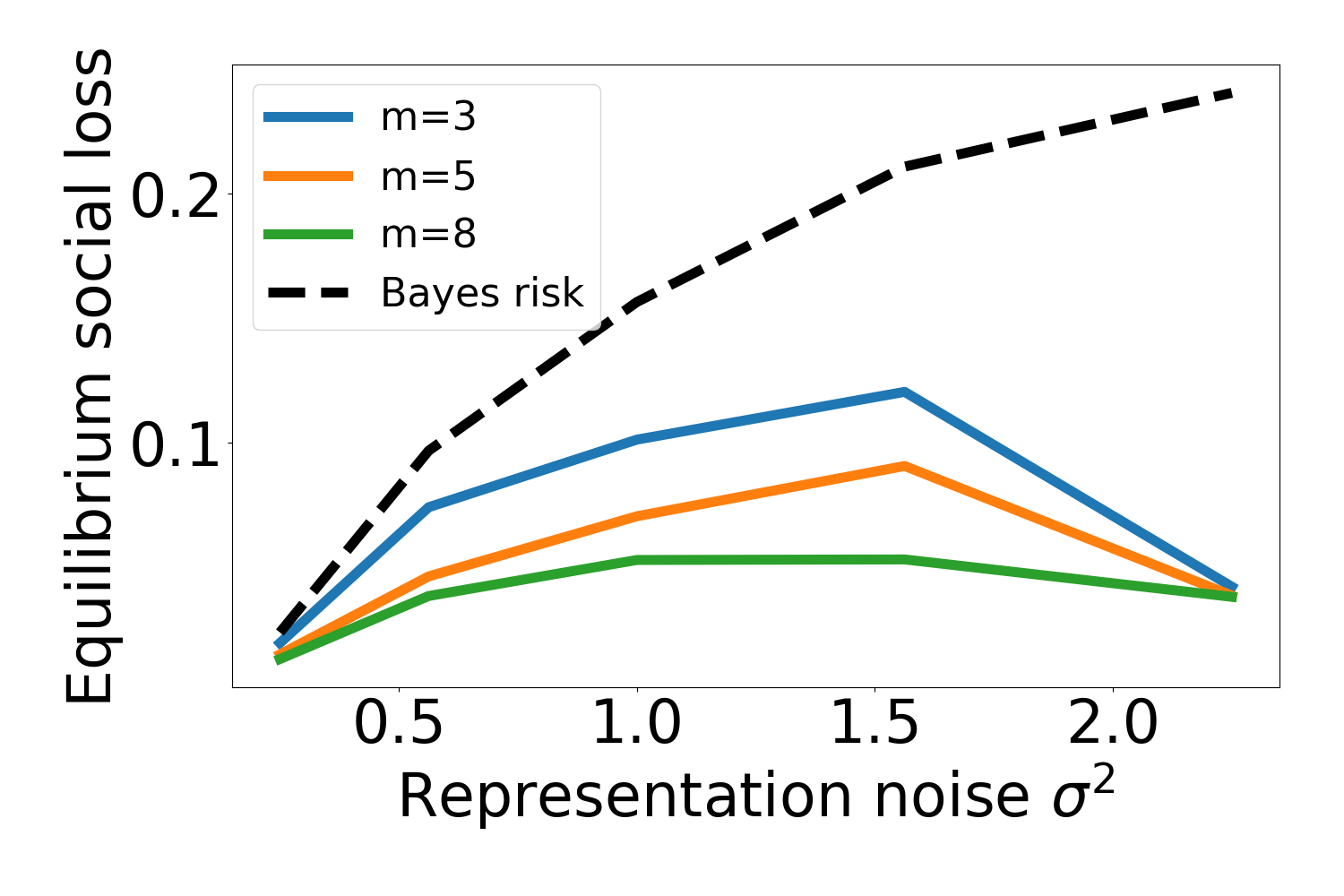}
        \caption{}
        \label{fig:simulations-noise}
    \end{subfigure}
    \begin{subfigure}[b]{0.32\textwidth}
        \includegraphics[scale=0.12]{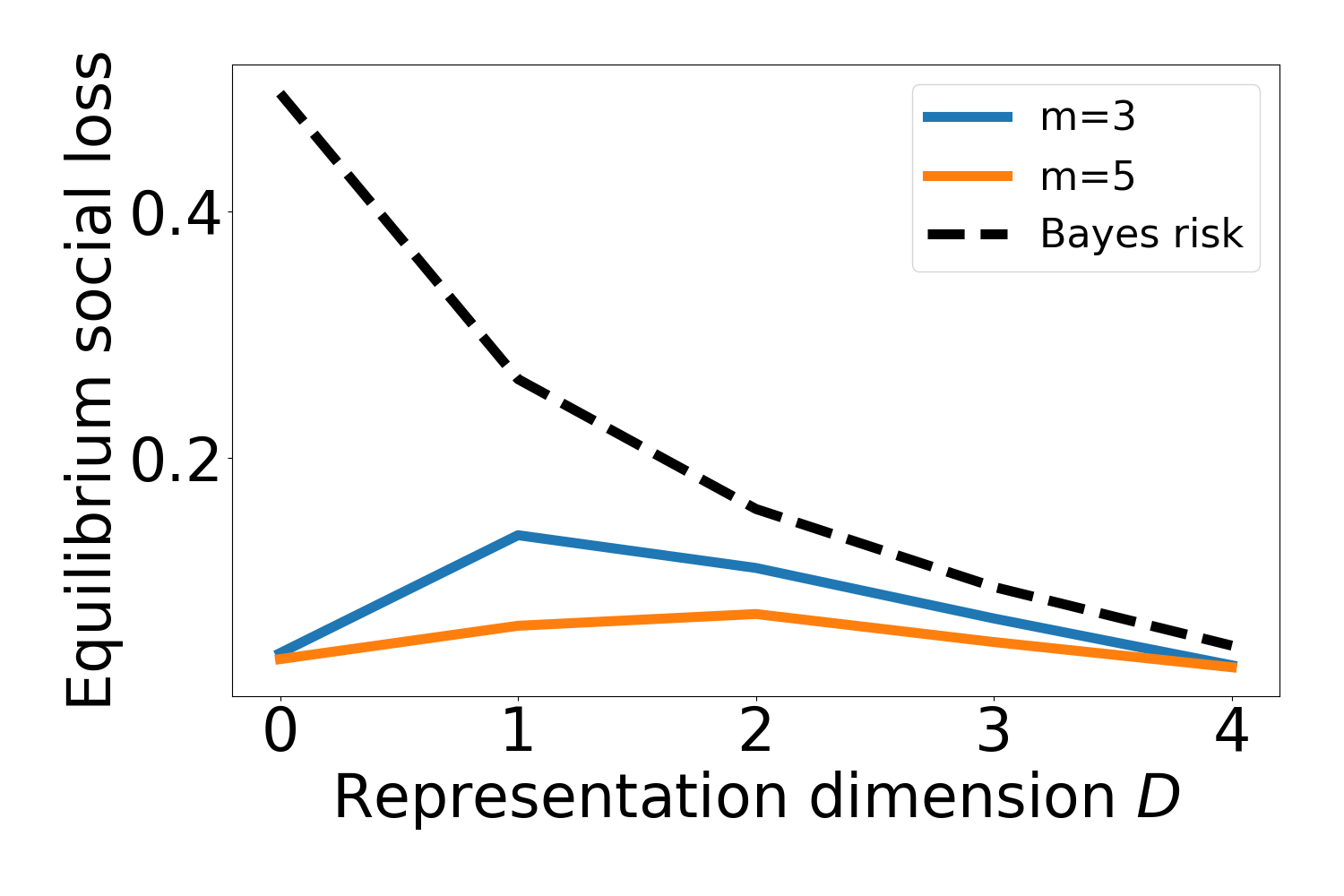}
        \caption{}
        \label{fig:simulations-dimension}
    \end{subfigure}
    \caption{Equilibrium social loss (y-axis) versus data representation quality (x-axis) given $m$ model-providers, for different function classes $\mathcal{F}$ (rows) and when representations are varied along different aspects (columns). Top row: $\mathcal{F} = \mathcal{F}^{\text{binary}}_{\text{all}}$, with closed-form formula from Proposition \ref{prop:idealized}. Bottom row: linear functions, computed via simulation (Section \ref{sec:linear}). We vary representations with respect to per-representation Bayes risk (a,d), noise level (b,e), and dimension (c,f). The dashed line indicates the Bayes risk (omitted if it is too high to fit on the axis). The Bayes risk is monotone, but the equilibrium social loss is non-monotone.}
    \label{fig:main}
\end{figure}

Using Proposition~\ref{prop:idealized}, we next vary representations along several axes and compute the equilibrium social loss, observing non-monotonicity in each case.

\paragraph{Setting 1: Varying the per-representation Bayes risks.} Consider a population with a single value of $x$ that has Bayes risk $\alpha(x) = \alpha$. We vary representation quality by varying $\alpha$ from $0$ to $0.5$. 
Figure \ref{fig:theory-bayes} depicts the result: by Proposition \ref{prop:idealized}, the equilibrium social loss is zero if $\alpha >1/m$ and is $\alpha$ if $\alpha < 1/m$, leading to  
non-monotonicity at $\alpha = 1/m$. When there are $m \ge 3$ model-providers, the equilibrium social loss is thus non-monotonic in $\alpha$. (For $m = 2$, where $\alpha = 1/2$ is the maximum possible per-representation Bayes risk, the equilibrium social loss is monotone in $\alpha$.) As the number of model-providers increases, the non-monotonicity occurs at a higher data representation quality (a lower Bayes risk). 

\paragraph{Setting 2: Varying the representation noise.} Consider a one-dimensional population given by a mixture of two Gaussians (one for each class), where each Gaussian has variance $\sigma^2$ (see Appendix \ref{appendix:setup} for the details of the setup). We vary the parameter $\sigma$ to change the representation quality. Intuitively, a lower value of $\sigma$ makes the Gaussians more well-separated, which improves representation quality (Bayes risk). By Proposition \ref{prop:idealized}, the equilibrium social loss is $\E \left[\alpha(x) \cdot \mathbbm{1}[\alpha(x) < 1/m] \right]$. For each value of $\sigma$, we estimate the equilibrium social loss by sampling representations $x$ from the population and taking an average.\footnote{Strictly speaking, we can't directly apply Proposition \ref{prop:idealized} to this setup since $X$ is infinite. We circumvent this issue by applying Proposition \ref{prop:idealized} on a sample of the representations.} 
Figure \ref{fig:theory-noise} depicts the result: the equilibrium social loss is non-monotonic in $\sigma$ (and thus the Bayes risk). Again, as the number of model-providers increases, the non-monotonicity occurs at a higher representation quality (a lower Bayes risk).  

\paragraph{Setting 3: Varying the representation dimension.} We consider a four-dimensional population $(X^{\text{all}}, Y)$, and let the representation $X$ consist of the first $D$ coordinates of $X^{\text{all}}$, for $D$ varying from $0$ to $4$ (see Appendix \ref{appendix:setup} for full details). Intuitively, a higher dimension $D$ makes the representations more informative, thus improving representation quality (Bayes risk). As before, for each value of $D$, we estimate the equilibrium social loss by sampling representations $x$ from the population and taking an average. Figure \ref{fig:theory-dimension} depicts the result: the equilibrium social loss is once again non-monotonic in the representation dimension $D$ (and thus the Bayes risk). 

\paragraph{Discussion.} 
Settings 1-3 illustrate that equilibrium social loss can be non-monotonic in Bayes risk when representations are improved along many qualitatively different axes. The intuition is that varying representations along these axes can increase the values of $\alpha(x)$ for inputs $x$; by Proposition \ref{prop:idealized}, these changes to $\alpha(x)$ can lead to non-monotonicity in the equilibrium social loss. We will revisit Settings 1-3 for richer market structures (Section \ref{subsec:unequal}) and for linear predictors and noisy user decisions (Section \ref{subsec:synthetic}). 

\subsection{Generalization to more than 2 classes}\label{subsec:idealizedmulticlass}

While our analysis has thus far focused on classification with $K = 2$ classes, the number of classes $K$ can be much larger in practice. As a motivating example, consider content recommendation tasks where each class represents a different genre of content; since the content landscape can be quite diverse, we would expect $K$ to be fairly large.\footnote{When $K$ is large, even if users can ``search'' for and ``consume'' content on their own without relying on model-provider predictions, we expect that our measure of social loss would still be a good proxy for the loss experienced by users. In particular, it would be prohibitively expensive for users to try out all $K$ classes, so classes that are not suggested to the user by any model-provider's predictions might be effectively inaccessible to the user.} This motivates us to extend our theoretical characterization in Proposition \ref{prop:idealized} to classification with $K \ge 2$ classes.

For the case of $K \ge 2$ classes, the appropriate  analogue of the per-representation Bayes risk is the \text{per-class-per-representation Bayes risk}, defined to be:
\begin{equation}
\label{eq:alphabayesclass}
 \alpha^i(x) := \Prob(y = i \mid x) 
\end{equation}
for each $x \in X$ and $i \in \left\{0,1,\ldots, K-1\right\}$. Observe that $1 - \max_{0 \le i \le K-1} \alpha^i(x)$ is the minimum error that a single model-provider can hope to achieve on $x$, and $\texttt{OPT}_{\text{single}}$ is equal to the average value $\E[1 - \max_{0 \le i \le K-1} \alpha^i(x)]$ across the population. The equilibrium social loss, however, depends on other aspects of the $\alpha^i(x)$ values. 

We characterize the equilibrium social loss in terms of the per-class-per-representation Bayes risks in the following proposition. Our characterization again focuses on  pure-strategy equilibria, which are guaranteed to exist in this setup by Proposition \ref{prop:existence}.
\begin{proposition}
\label{prop:idealizedmulticlass}
Let $X$ be a finite set, let there be $K \ge 2$ classes, let $\mathcal{F} = \mathcal{F}^{\text{multi-class}}_{\text{all}}$. Suppose that user decisions are noiseless (i.e., user decisions are given by \eqref{eq:userchoicespecific}). Let $c = \min_{x \in X} \max_{0 \le i \le K-1} \alpha^i(x)$. Then, at any pure strategy Nash equilibrium $f^*_1, \ldots, f^*_m$, the social loss $\texttt{SL}(f^*_1, \ldots, f^*_m)$ is bounded as
\begin{equation}
\label{eq:sociallossmulticlass} 
\E_{(x,y) \sim \mathcal{D}} \left[\sum_{i=1}^K \alpha^i(x) \cdot \mathbbm{1}\left[\alpha^i(x) < \frac{c}{m}\right] \right] \le \texttt{SL}(f^*_1, \ldots, f^*_m) \le  \E_{(x,y) \sim \mathcal{D}} \left[\sum_{i=1}^K \alpha^i(x) \cdot \mathbbm{1}\left[\alpha^i(x) \le \frac{1}{m} \right] \right]. 
\end{equation}
\end{proposition}

The high-level intuition for Proposition \ref{prop:idealizedmulticlass} is similar to the intuition for Proposition \ref{prop:idealized}, except that each class needs to be considered separately. In particular, when class $i$ occurs sufficiently frequently for the representation $x$ (i.e., when $\alpha^i(x)$ is not too small), then some model-provider will label $x$ as $i$; on the other hand, if the class $i$ occurs very infrequently for $x$, then no model-provider will label $x$ as $i$. We defer a proof of Proposition \ref{prop:idealizedmulticlass} to Appendix \ref{appendix:theory}.

While Proposition \ref{prop:idealizedmulticlass} is conceptually a generalization of Proposition \ref{prop:idealized}, the details of Proposition \ref{prop:idealizedmulticlass} slightly differ. In particular, Proposition \ref{prop:idealizedmulticlass} does not completely pin down the equilibrium social loss, and there is a factor of $c$ slack in the constraint on each $\alpha^i(x)$ in \eqref{eq:sociallossmulticlass} between the upper and lower bounds. Nonetheless, since the value $c = \min_{x \in X} \max_{0 \le i \le K-1} \alpha^i(x)$ measures the minimum accuracy of the Bayes optimal predictor across all inputs $x$, we expect that ``reasonable'' representations (i.e., representations which are sufficiently informative) would have $c$ equal to a constant. When $c$ is a constant, there is at most a constant factor slack in the $\alpha^i(x)$ constraints in \eqref{eq:sociallossmulticlass} between upper and lower bound. 

For similar reasons to Proposition \ref{prop:idealized}, Proposition \ref{prop:idealizedmulticlass}  implies that the equilibrium social loss can be non-monotonic in the representation quality (i.e., the Bayes risk). As a concrete example, consider the following adaptation of Setting 1 in Section \ref{subsec:idealized}: let there be a population with a single value of $x$ where $\alpha_0(x) = 1 - 2\alpha$, $\alpha_1(x) = \alpha$, and $\alpha_2(x) = \alpha$ for some $\alpha < 1/4$. In this setup, we see that $c \ge 1/2$. By Proposition \ref{prop:idealizedmulticlass}, the equilibrium social loss is $2 \alpha$ if $\alpha < 1/(2m)$, and the equilibrium social loss is $0$ if $\alpha > 1/m$; on the other hand, the Bayes risk is equal to $2 \alpha$ for any $\alpha < 1/4$. This illustrates that the equilibrium social loss is non-monotonic in the Bayes risk. We expect that other setups similar to those in Section \ref{subsec:idealized} will also lead to non-monotonicity for multi-class tasks.

\subsection{Generalization to unequal market reputations}\label{subsec:unequal}
While we assumed above that users evenly break ties between model-providers, in reality, users might be more likely to choose model-providers with a higher market reputation (e.g., established, popular model-providers). This motivates us to incorporate market reputations into user decisions. 

Formally, we assign to each model-provider $j$ a \textit{market reputation} $w_j$, and we replace the logit model in \eqref{eq:userchoice} with a {weighted} logit variant. 
When $c\to 0$, rather than breaking ties uniformly, they are instead broken proportionally to $w_j$: 
\begin{equation}
    \label{eq:userdecisionsmr}
    \Prob[j^*(x,y) = j] = 
\begin{cases}
0 & \text{ if } j \not\in \argmin_{j' \in [m]} \mathbbm{1}[y \neq f_{j'}(x)] \\
\frac{w_j}{\sum_{j'' \in [m]} w_{j''} \cdot \mathbbm{1}[j'' \in \argmin_{j' \in [m]} \mathbbm{1}[y \neq f_{j'}(x)]]}
& \text{ if } j \in \argmin_{j' \in [m]} \mathbbm{1}[y \neq f_{j'}(x)].
\end{cases}
\end{equation} 
See Section \ref{sec:modeldiscussion} for further discussion of this model. For simplicity, we assume that market reputations are normalized to sum to one. 

Similarly to Proposition \ref{prop:idealized}, we derive a closed-form formula for the equilibrium social loss, focusing on the case of binary classification with $m = 2$ model-providers for analytic tractability. We observe non-monotonicity as before, but with a more complex functional form. 
\begin{proposition}
\label{prop:2model-providers}
Let $X$ be a finite set, let $K = 2$, and let $\mathcal{F} = \mathcal{F}^{\text{binary}}_{\text{all}}$. Suppose there are $m=2$ model-providers with market reputations $w_{\text{min}}$ and $w_{\text{max}}$, where $w_{\text{max}} \ge w_{\text{min}}$ and $w_{\text{max}} +  w_{\text{min}} = 1$. Suppose that user decisions are given by \eqref{eq:userdecisionsmr}, and that $\alpha(x) \neq w_{\text{min}}$ for all $x \in X$.\footnote{As with Proposition~\ref{prop:idealized}, when $\alpha(x)$ is equal to $w_{\text{min}}$ for some value of $x$, there are multiple equilibria.} At any (mixed) Nash equilibrium $(\mu_1, \mu_2)$, the expected social loss  is equal to: 

\begin{equation}
\label{eq:2model-providerssocialloss} 
\E_{\substack{f_1 \sim \mu_1 \\ f_2 \sim \mu_2}}[\texttt{SL}(f_1, f_2)] = \E_{(x, y) \sim \mathcal{D}} \left[\underbrace{\frac{(\alpha(x) - w_{\text{min}}) \cdot (w_{\text{max}}  - \alpha(x))}{(1 - 2 \cdot w_{\text{min}})^2}}_{(A)} \cdot \mathbbm{1}[\alpha(x) > w_{\text{min}} ]  + \underbrace{\alpha(x)}_{(B)} \cdot \mathbbm{1}[\alpha(x) < w_{\text{min}} ]  \right]. 
\end{equation}
\end{proposition}

\normalsize{}

The high-level intuition for Proposition \ref{prop:2model-providers}, like for Proposition \ref{prop:idealized}, is that the equilibrium predictions go from heterogeneous to homogenous as $\alpha(x)$ decreases. Term (A), which is realized for large $\alpha(x)$, captures the equilibrium social loss for heterogeneous predictions. Term (B), which is realized for small $\alpha(x)$, captures the equilibrium social loss for homogeneous predictions. We defer the proof of Proposition \ref{prop:2model-providers} to Appendix \ref{appendix:theory}.

\begin{figure}[t]
    \centering
    \begin{subfigure}[b]{0.32\textwidth}
        \includegraphics[scale=0.12]{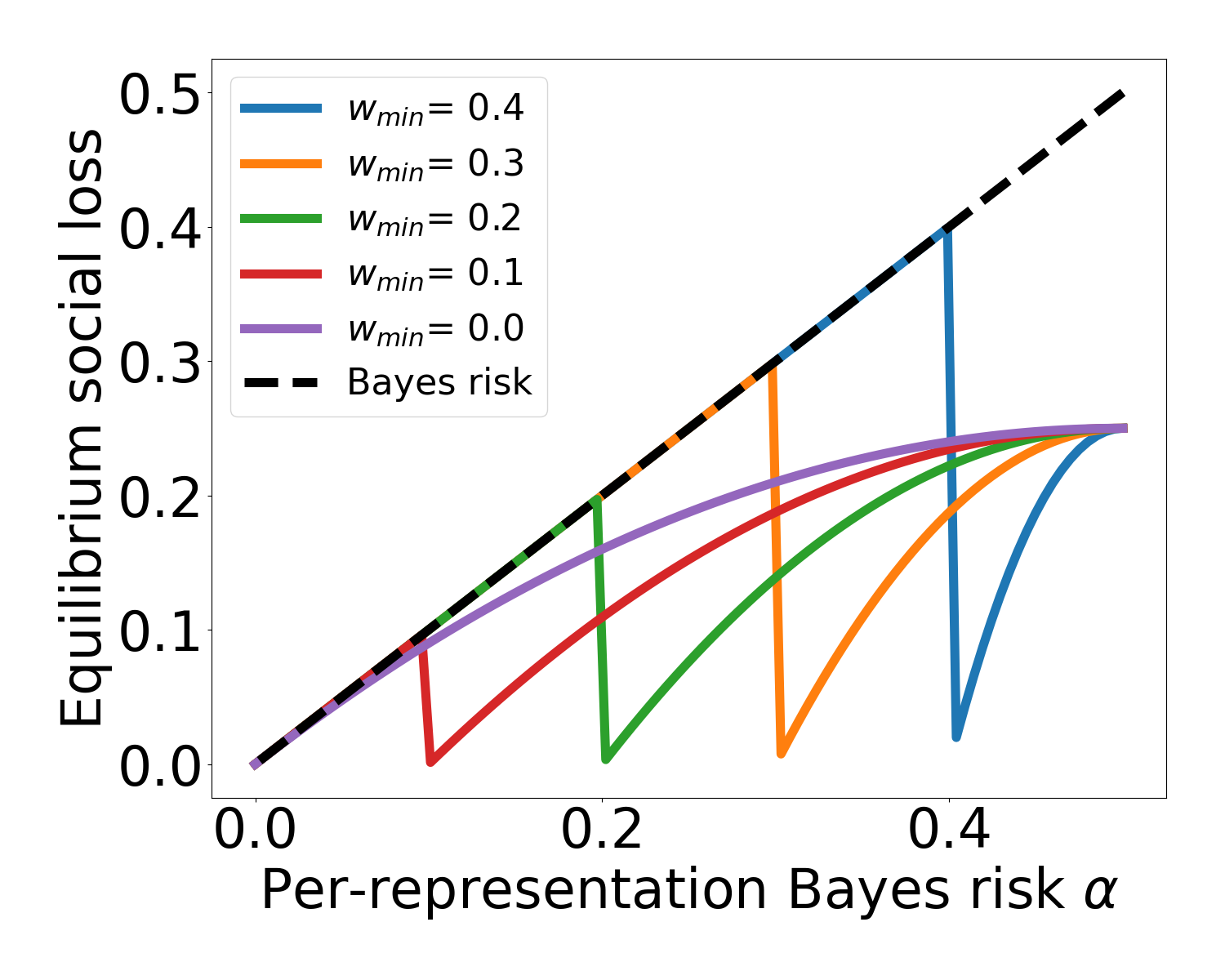}
        \caption{}
        \label{fig:theory-bayes-varyw}
    \end{subfigure}
    \begin{subfigure}[b]{0.32\textwidth}
        \includegraphics[scale=0.12]{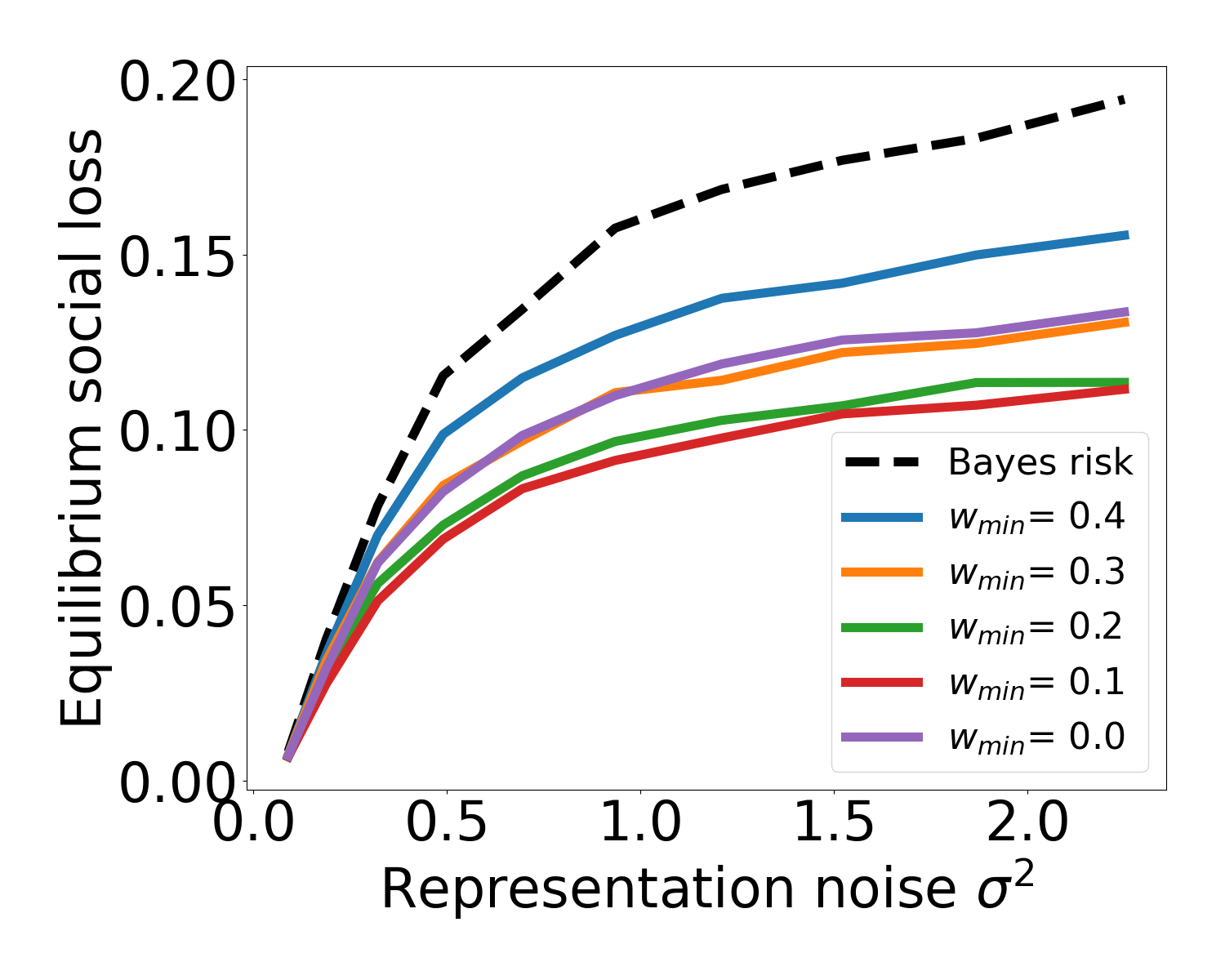}
        \caption{}
        \label{fig:theory-noise-varyw}
    \end{subfigure}
    \begin{subfigure}[b]{0.32\textwidth}
        \includegraphics[scale=0.12]{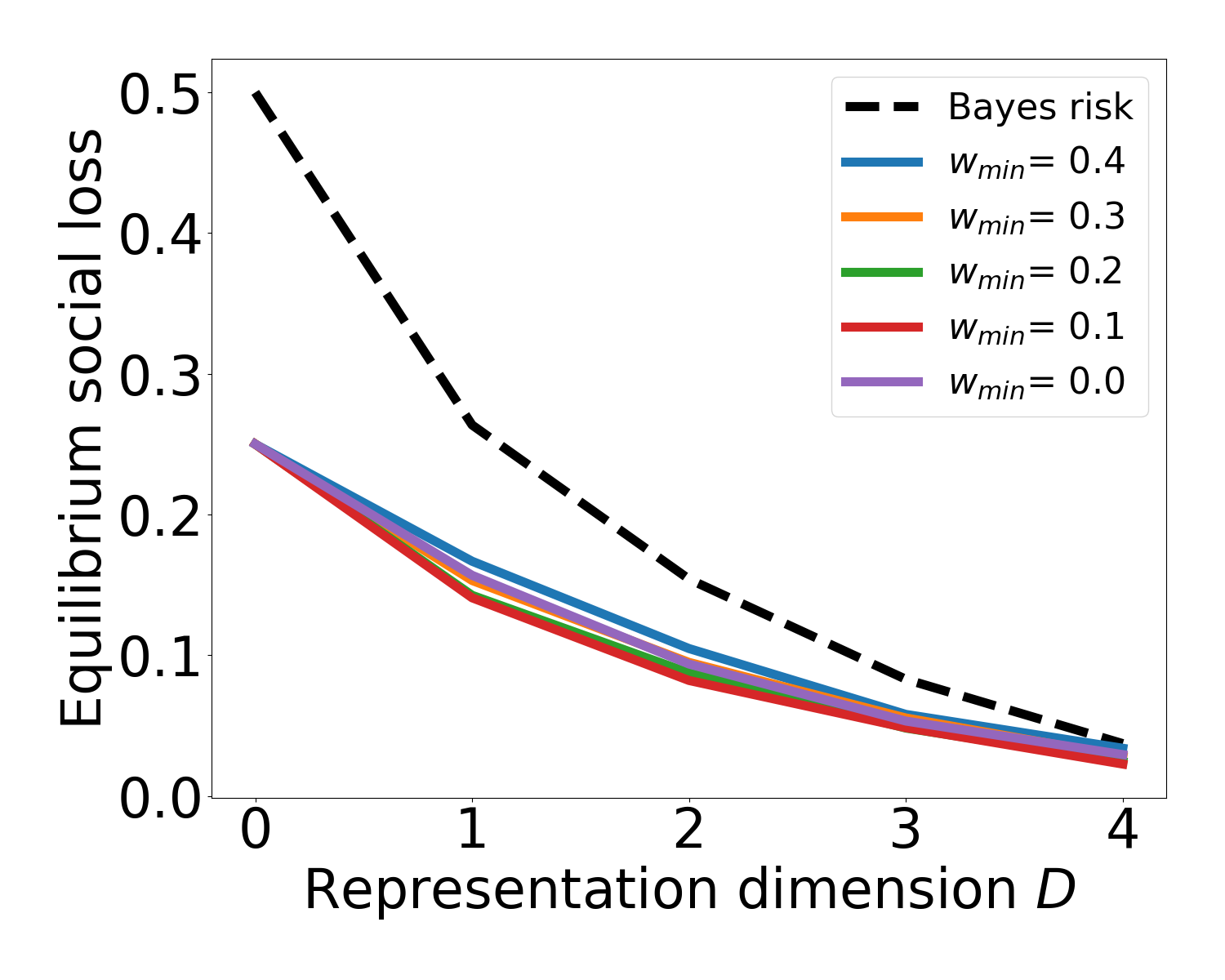}
        \caption{}
        \label{fig:theory-dimension-varyw}
    \end{subfigure}
    \caption{Equilibrium social loss (y-axis) versus data representation quality (x-axis) given two model-providers with market reputations $[1 - w_{\text{min}}, w_{\text{min}}]$ when representations are varied along different aspects (columns). The equilibrium social loss is computed via the closed-form formula from Proposition \ref{prop:2model-providers}. We vary representations with respect to per-representation Bayes risk (a), noise level (b), and dimension (c). The dashed line indicates the Bayes risk. The Bayes risk is monotone for all 3 axes of varying representations; on the other hand, the equilibrium social loss is non-monotone in the per-representation Bayes risk and monotone in noise level and dimension.}
    \label{fig:2model-providers}
\end{figure}

The details of Proposition \ref{prop:2model-providers} differ from Proposition \ref{prop:idealized} in several ways. First, the transition point from heterogeneous to homogeneous predictions occurs at $\alpha(x) = w_{\text{min}}$ as opposed to $\alpha(x) = 1/2$. In particular, the transition point depends on the market reputations rather than only the number of model-providers. Second, the equilibria have \textit{mixed strategies} rather than pure strategies, because pure-strategy equilibria do not necessarily exist when market reputations are unequal (see Lemma \ref{lemma:nonexistence} in Appendix \ref{appendix:theory}). Third, the social loss at a representation $x$ is no longer equal to zero for heterogeneous predictions---in particular, term (A) is now positive for all $\alpha(x) >  w_{\text{min}}$ and increasing in $\alpha(x)$.

To better understand the implications of Proposition \ref{prop:2model-providers}, we revisit Settings 1-3 from Section \ref{subsec:examples}, considering the same three axes of varying representations with the same distributions over $(x,y)$. In contrast to Section \ref{subsec:examples}, we consider 2 competing model-providers with unequal market positions rather than $m$ competing model providers with equal market positions. Our results, described below, are depicted in Figure \ref{fig:2model-providers}. 

\paragraph{Setting 1: Varying the per-representation Bayes risks.} Consider the same setup as Setting 1 in Section \ref{subsec:examples}. Figure \ref{fig:theory-bayes-varyw} depicts the non-monotonicity of the equilibrium social loss in the per-representation Bayes risk $\alpha$ across different settings of market reputations for 2 competing model-providers. The discontinuity occurs at the smaller market reputation $w_{\text{min}}$. Thus, as the market reputations of the 2 model-providers become closer together, the non-monotonicity occurs at a lower data representation quality (higher Bayes risk). 

\paragraph{Settings 2-3: Varying the representation noise or representation dimension.} Consider the setups from Settings 2-3 in Section \ref{subsec:examples}.  Figures \ref{fig:theory-noise-varyw}-\ref{fig:theory-dimension-varyw} depicts that the equilibrium social loss is \textit{monotone} in data representation quality (Bayes risk) across different settings of market reputations for 2 competing model-providers.

\begin{figure}[t]
    \centering
    \begin{subfigure}{0.49\textwidth}
        \includegraphics[scale=0.32]{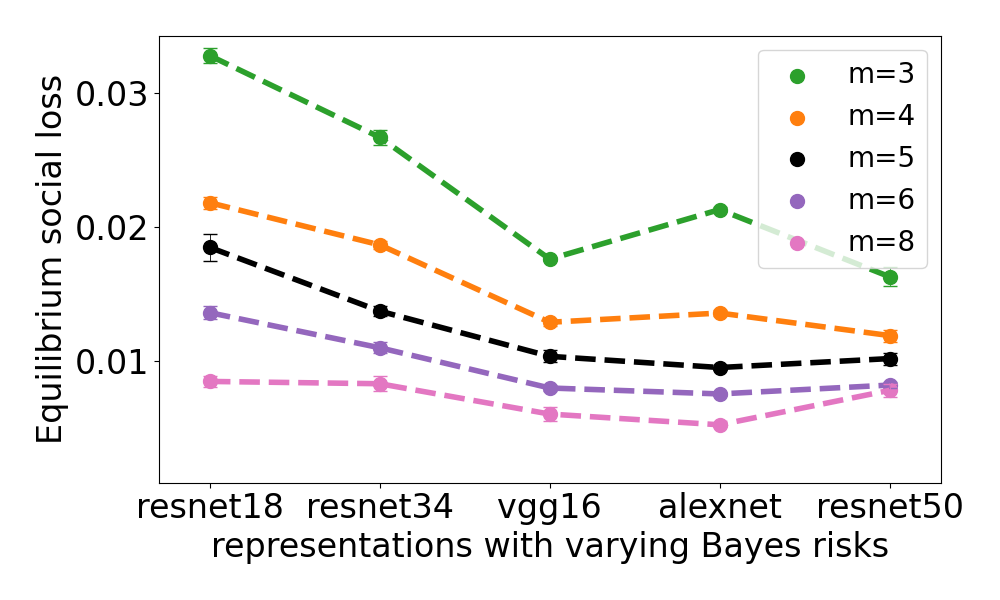}
        \caption{}
        \label{fig:CIFAR-10}
    \end{subfigure}
    \begin{subfigure}{0.49\textwidth}
        \includegraphics[scale=0.32]{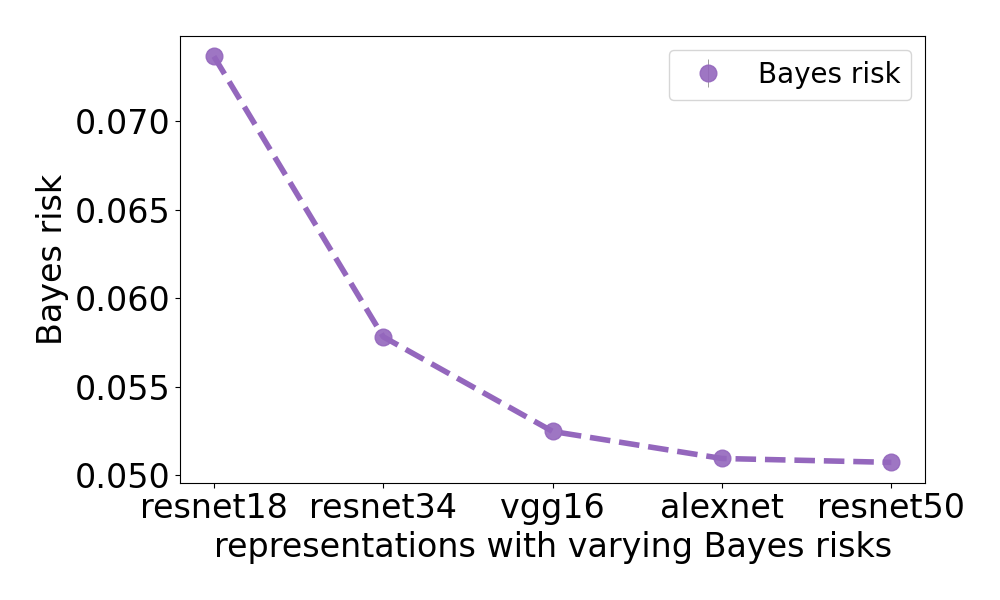}
        \caption{}
        \label{fig:bayesoptcifar}
    \end{subfigure}
    \caption{ Equilibrium social loss (left) and Bayes risk (right) on a binary classification task on CIFAR-10 (Section \ref{subsec:image}). Representations are generated from different networks pre-trained on ImageNet. The points show the equilibrium social loss when $m$ model-providers compete with each other (left) and the Bayes risk of a single model-provider in isolation (right). While Bayes risk is decreasing in this representation ordering, the equilibrium social loss is non-decreasing in this ordering. 
    The equilibrium social loss is thus non-monotonic in representation quality as measured by Bayes risk. Error bars are 1 standard error.
   }
    \label{fig:other}
\end{figure}

\paragraph{Discussion.} To interpret these results, observe that for 2 model-providers with equal market reputations ($w_{\text{min}} = 0.5$), the equilibrium social loss is always equal to the Bayes risk by Propositions \ref{prop:idealized}-\ref{prop:2model-providers}, which trivially implies monotonicity. In contrast, Figure \ref{fig:2model-providers} shows that for unequal market positions ($w_{\text{min}} < 0.5$), the equilibrium social loss is non-monotonic in Bayes risk for Setting 1, though it is still monotonic in Bayes risk for Settings 2 and 3. (For comparison, recall from Figures \ref{fig:theory-bayes}-\ref{fig:theory-dimension} that for $m \gg 2$ model-providers with equal market reputations, non-monotonicity was exhibited for all three settings.) An interesting open question is identify other axes of varying representations, beyond Setting 1, which lead to non-monotonicity for 2 model-providers with unequal market reputations.

\begin{figure}[t]
    \centering
    \begin{subfigure}{0.49\textwidth}
        \includegraphics[scale=0.32]{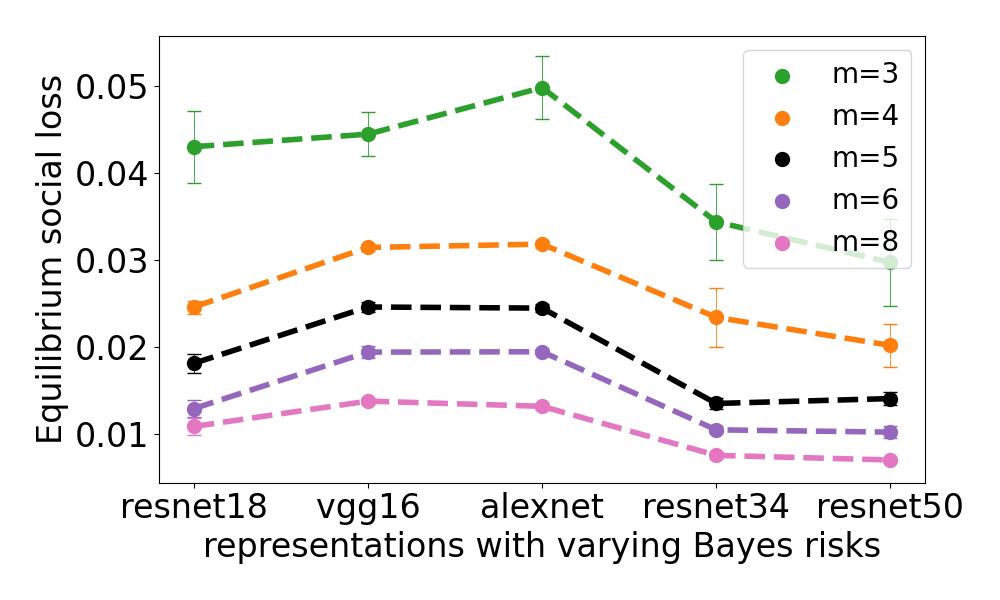}
        \caption{}
        \label{fig:CIFAR-1010class}
    \end{subfigure}
    \begin{subfigure}{0.49\textwidth}
        \includegraphics[scale=0.32]{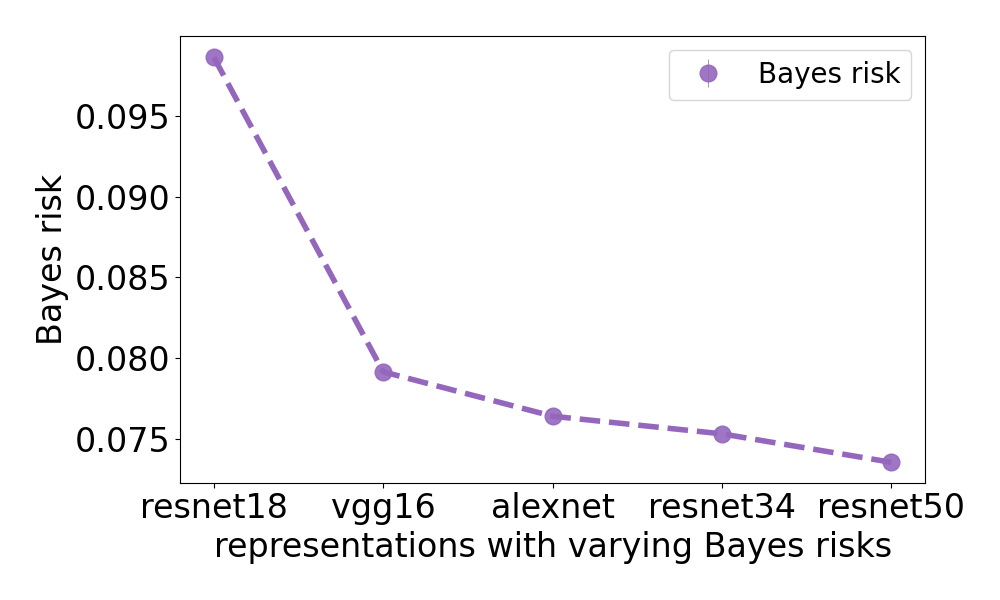}
        \caption{}
        \label{fig:bayesoptcifar10class}
    \end{subfigure}
    \caption{ Equilibrium social loss (left) and Bayes risk (right) on a 10-class classification task on CIFAR-10 (Section \ref{subsec:image10}). Representations are generated from different networks pre-trained on ImageNet. The points show the equilibrium social loss when $m$ model-providers compete with each other (left) and the Bayes risk of a single model-provider in isolation (right). While Bayes risk is decreasing in this representation ordering, the equilibrium social loss is non-decreasing in this ordering. 
    The equilibrium social loss is thus non-monotonic in representation quality as measured by Bayes risk. Error bars are 1 standard error.
   }
    \label{fig:10class}
\end{figure}

\section{Empirical Analysis of Non-monotonicity for Linear Predictors }\label{sec:linear}

We next turn to linear predictors and demonstrate empirically that the social welfare can be non-monotonic in data representation quality in this setup as well.\footnote{The code can be found at \url{https://github.com/mjagadeesan/competition-nonmonotonicity}.} We take $X = \mathbb{R}^D$ and we let the model parameters be $\phi$. For binary classification, we let $\mathcal{F}^{\text{binary}}_{\text{linear}}$ be the family of linear predictors $f_{w,b} = \sigmoid(\langle w, x\rangle +b)$ where $w \in \mathbb{R}^D$, $b \in \mathbb{R}$, and $\phi = [w, b]$. Similarly, for classification with more than 2 classes, we let $\mathcal{F}^{\text{multi-class}}_{\text{linear}}$ be the family of linear predictors $f_{W,b} = \softmax(Wx + b)$ where $w \in \mathbb{R}^{ |Y| \times D}$, $b \in \mathbb{R}^{|Y|}$, and $\phi = [W, b]$.  Since this setting no longer admits closed-form formulae, we numerically estimate the equilibria using a variant of \textit{best-response dynamics}, where model-providers repeatedly best-respond to the other predictors. 

We first show on low-dimensional synthetic data on a binary classification task that the insights from Section \ref{subsec:examples} readily generalize to linear predictors (see Figures \ref{fig:simulations-bayes}-\ref{fig:simulations-dimension}). We then turn to natural data, considering binary and 10-class image classification tasks for CIFAR-10 and using pretrained networks---AlexNet, VGG16, and various ResNets---to generate high-dimensional representations (ranging from 512 to 4096). In this setting we again find that the equilibrium social loss can be non-monotonic in the Bayes risk (see Figure \ref{fig:other} and Figure \ref{fig:10class}).

\subsection{Best-response dynamics implementation}
\label{subsec:bestresponsedynamics}

To enable efficient computation, we assume the distribution $\mathcal{D}$ corresponds to a finite dataset with $N$ data points. We calculate equilibria using an approximation of best-response dynamics. Model-providers (players) iteratively (and approximately) best-respond to the other players' actions. We implement the approximate best-response as running several steps of gradient descent.  

In more detail, for each $j \in [m]$, we initialize the model parameters $\phi$ as mean zero Gaussians with standard deviation $\sigma$. Our algorithm then proceeds in stages. At a given stage, we iterate through the model-providers in the order $1, \ldots, m$. When $j$ is chosen, first we decide whether to reinitialize: if the risk $\mathbb{E}_{(x,y) \sim \mathcal{D}}[\ell(f_{\phi}(x), y)]$ exceeds a threshold $\rho$, we re-initialize $w_j$ and $b_j$ (sampling from mean zero Gaussians as before); otherwise, we do not reinitialize. Then we run gradient descent on $u(\cdot; \mathbf{f}_{-j})$ (computing the gradient on the full dataset of $N$ points) with learning rate $\eta$ for $I$ iterations, updating the parameters $\phi$. We run this gradient descent step up to 2 more times if the risk $\mathbb{E}_{(x,y) \sim \mathcal{D}}[\ell(f_{\phi}(x), y)]$ exceeds a threshold $\rho'$. At the end of a stage, the stopping condition is that for every $j \in [m]$, model-provider $j$'s utility $u(f_j, \mathbf{f}_{-j})$ has changed by at most $\epsilon$ relative to the previous stage. If the stopping condition is not met, we proceed to the next stage. 

\subsection{Simulations on synthetic data}\label{subsec:synthetic}

We first revisit Settings 1-3 from Section \ref{subsec:examples}, considering the same three axes of varying representations with the same distributions over $(x,y)$. In contrast to Section~\ref{subsec:examples}, we restrict the model family to linear predictors $\mathcal{F}^{\text{binary}}_{\text{linear}}$ instead of allowing all predictors $\mathcal{F}^{\text{binary}}_{\text{all}}$. We also set the noise parameter $c$ in user decisions \eqref{eq:userchoice} to $0.3$. Our goal is to examine if the findings from Section \ref{sec:theory} generalize to this new setting. 

We compute the equilibria for each of the following (continuous) distributions as follows. First, we let $\mathcal{D}$ be the empirical distribution over $N = 10,000$ samples from the continuous distribution. Then we run the best-response dynamics described in Section \ref{subsec:bestresponsedynamics} with $\rho = 0.3$, $I = 5000$, $\eta = 0.1$, $\epsilon = 0.01$, and $\sigma = 0.1$. 
We then compute the equilibrium social loss according to \eqref{eq:sociallossdef}.  We also compute the Bayes optimal predictor with gradient descent. See Appendix \ref{appendix:setup} for full details. 
 
Our results, described below, are depicted in Figures \ref{fig:simulations-bayes}-\ref{fig:simulations-dimension} (row 2). We compare these results with Figures \ref{fig:theory-bayes}-\ref{fig:theory-dimension} (row 1), which shows the analogous results for $\mathcal{F}^{\text{binary}}_{\text{all}}$ from Section \ref{subsec:examples}. 

\paragraph{Setting 1: Varying the per-representation Bayes risks.} Consider the same single $x$ setup as in Setting 1 in Section \ref{subsec:examples}. The only parameter of the predictor is the bias $b \in \mathbb{R}$ (i.e., we treat $x$ as zero-dimensional). 
Figure \ref{fig:simulations-bayes} shows that the equilibrium social loss is non-monotonic in $\alpha$, which mirrors the non-monotonicity in Figure \ref{fig:theory-bayes}.

\paragraph{Setting 2: Varying the representation noise.} 
Consider the same one-dimensional mixture-of-Gaussians distribution as in Setting 2 in Section \ref{subsec:examples}. (The weight $w$ is one-dimensional.) We again vary the noise $\sigma$ to change the representation quality. 
Figure \ref{fig:simulations-noise} shows that the equilibrium social loss is non-monotonic in the noise $\sigma$, which again mirrors the non-monotonicity in  Figure \ref{fig:theory-noise}. 

\paragraph{Setting 3: Varying the representation dimension.} Consider the same four-dimensional population as in Setting 3 in Section \ref{subsec:examples}. We vary the representation dimension $D$ from $0$ to $4$ to change the representation quality. Figure \ref{fig:simulations-dimension} shows that the equilibrium social loss is non-monotonic in the dimension $D$, which once again mirrors the non-monotonicity in  Figure \ref{fig:theory-dimension}. 

\paragraph{Discussion.}
In summary, in Figure \ref{fig:main}, rows 1 and 2  exhibit similar non-monotonicities. This illustrates that the insights from Section \ref{subsec:idealized} translate to linear predictors and noisy user decisions. 

\subsection{Simulations on CIFAR-10 for binary classification}\label{subsec:image}

We next turn to experiments with natural data. While we have directly varied the informativeness of data representations thus far, representations in practice are frequently generated by pretrained models. The choice of the pretrained model implicitly influences representation quality, as measured by Bayes risk on the downstream task. In this section, we consider how the equilibrium social loss changes with representations generated from pretrained models of varying quality. We restrict the model family to linear predictors $\mathcal{F}^{\text{binary}}_{\text{linear}}$ and set the noise parameter $c$ in user decisions \eqref{eq:userchoice} to $0.1$.

We consider a binary image classification task on CIFAR-10 \citep{CIFAR} with 50,000 images. Class $0$ is defined to be $\left\{\text{airplane}, \text{bird}, \text{automobile}, \text{ship}, \text{horse}, \text{truck}\right\}$ and the class $1$ is defined to be $\left\{\text{cat}, \text{deer}, \text{dog}, \text{frog}\right\}$. We treat the set of 50,000 images and labels as the population of users, meaning that it is both the training set and the validation set.\footnote{We make this choice to be consistent with the rest of the paper, where we focus on population-level behavior and thus do not consider  generalization error.} Representations are generated from five models---AlexNet \citep{AlexNet}, VGG16 \citep{VGGref}, ResNet18, ResNet34, and ResNet50 \citep{ResNet}---pretrained on ImageNet \citep{ImageNet}. The representation dimension is 4096 for AlexNet and VGG16, 512 for ResNet18 and ResNet34, and 2048 for ResNet50. 

We compute the equilibria as follows. First, we let $\mathcal{D}$ be the distribution described above with $N = 50,000$ data points. Then we run the best-response dynamics described in Section \ref{subsec:bestresponsedynamics} for $m \in \left\{3, 4, 5, 6, 8\right\}$ model-providers with $\rho = \rho' = 0.3$, $I = 2000$, $\epsilon = 0.001$, $\sigma = 0.5$, and a learning rate schedule that starts at $\eta = 1.0$.  We then compute the equilibrium social loss according to \eqref{eq:sociallossdef}.  We also compute the Bayes risk using gradient descent. For full experimental details, see Appendix \ref{appendix:setup}.

Figure \ref{fig:other} shows that the equilibrium social loss can be non-monotone in the Bayes risk. For example, for $m = 3$, VGG16 outperforms AlexNet, even though the Bayes risk of VGG16 is substantially higher than the Bayes risk of AlexNet. Interestingly, the location of the non-monotonicity differs across different values of $m$. For example, for $m = 5$ and $m = 8$, AlexNet outperforms ResNet50 despite having a higher Bayes risk, but ResNet50 outperforms AlexNet for $m = 3$ and $m = 4$.

\subsection{Simulations on CIFAR-10 for 10-class classification}\label{subsec:image10}

While our empirical analysis has thus far focused on binary classification, we now turn to classification with more than 2 classes. In particular, we consider a ten class CIFAR-10 \citep{CIFAR} task with 50,000 images. The labels are specified by the CIFAR-10 classes in the original dataset. We treat the set of 50,000 images and labels as the population of users, meaning that it is both the training set and the validation set. Representations are generated from the same five models---AlexNet \citep{AlexNet}, VGG16 \citep{VGGref}, ResNet18, ResNet34, and ResNet50 \citep{ResNet}---pretrained on ImageNet \citep{ImageNet}. We restrict the model family to linear predictors $\mathcal{F}^{\text{multi-class}}_{\text{linear}}$ and again set the noise parameter $c$ in user decisions \eqref{eq:userchoice} to $0.1$.

We compute the equilibria as follows. First, we let $\mathcal{D}$ be the distribution described above with $N = 50,000$ data points. Then we run the best-response dynamics described in Section \ref{subsec:bestresponsedynamics} for $m \in \left\{3, 4, 5, 6, 8\right\}$ model-providers with $\rho = 0.7$, $\rho' = 1.0$, $I = 2000$, $\epsilon = 0.001$, $\sigma = 0.5$, and a learning rate schedule that starts at $\eta = 1.0$. As before, we compute the equilibrium social loss according to \eqref{eq:sociallossdef}, and we also compute the Bayes risk using gradient descent. For full experimental details, see Appendix \ref{appendix:setup}.

Figure \ref{fig:10class} shows that the equilibrium social loss can be non-monotone in the Bayes risk. For example, across all five values of $m$, ResNet18 outperforms VGG16, even though the Bayes risk of ResNet is substantially higher than the Bayes risk of VGG16. Furthermore, for $m = 3$, VGG16 outperforms AlexNet despite having a larger Bayes risk. Interestingly, the shape of the equilibrium social loss curve for each value of $m$ (Figure \ref{fig:CIFAR-1010class}) appears qualitatively different than the analogous equilibrium social loss curve for binary classification (Figure \ref{fig:CIFAR-10}).

\section{Discussion of Model Assumptions}\label{sec:modeldiscussion}  

We highlight and discuss several assumptions that we make in our stylized model. 

\subsection{Assumptions on user decisions}

Our primary model for user decisions given by \eqref{eq:userchoice} is the standard logit model for discrete choice decisions \citep{train2002discrete} which is also known as the Boltzmann rationality model. In the limit as $c \rightarrow 0$, a user with representation $x$ and label $y$ select from the set of model-providers $\argmin_{j \in [m]} \ell(f_j(x), y)$ that achieve the minimum loss; in particular, the user chooses a model-provider from this set with probability proportional to the model-provider's market reputation. For $c > 0$, the specification in equation \eqref{eq:userchoice} captures that users evaluate a model-provider based on a noisy perception of the loss. 

While this model implicitly assumes that a user’s choice of platform is fully specified by the platforms’ choices of predictor (i.e. platforms are ex-ante homogeneous), we extend this model in Section \ref{subsec:unequal} to account for uneven market reputations across decisions. These market reputations are modeled as global weights in the logit model for discrete choice. 
Given market reputations $w_1, \ldots, w_m$, users choose a predictor according to: 
\begin{equation}
\label{eq:userchoicegeneralized}
\Prob[j^*(x,y) = j] = \frac{w_j \cdot e^{- \ell(f_j(x), y) / c}}{\sum_{j'=1}^m w_{j'} \cdot e^{- \ell(f_{j'}(x), y) / c}}.
\end{equation}
When the market reputations are all equal ($w_1 = \ldots = w_m$), equation \eqref{eq:userchoicegeneralized} exactly corresponds to \eqref{eq:userchoice}.  When the market reputations $w_j$ are not equal, equation \eqref{eq:userchoicegeneralized} captures that users place a higher weight on model-providers with a higher market reputation. This captures that users are more likely to choose a popular model-provider than a very small model-provider without much reputation. However, this formalization does assume that market reputations are global across users and that market reputations surface as tie-breaking weight in the noiseless limit.

Implicit in this model is asymmetric information between the model-providers and users. While the only information that a model-provider has about users is their representations, a user can make decisions based on noisy perceptions of their own loss (which can depend on their label). This captures that, even if users are unlikely to know their own labels,  users can experiment with multiple model-providers to (noisily) determine which one maximizes their utility. The inclusion of market reputations reflects that users are more likely to experiment with and ultimately choose popular model-providers than less popular model-providers. 

\subsection{Assumption of global data representations}

Our results assume that all model-providers share the same representations $x$ for each user and thus improvements in representations $x$ are experienced by all model-providers. This assumption is motivated by emerging marketplaces where different model-providers utilize the same pretrained model, but \textit{finetune} the model in different ways. To simplify this complex training process, we conceptualize pretraining as \textit{learning data representations (e.g., features)} and fine-tuning as \textit{learning a predictor from these representations}. In this formalization, increasing the scale of the pretrained model (e.g., by increasing the number of parameters or the amount of data) leads to improvements in data representations accessible to all of the model-providers during ``fine-tuning''.

An interesting direction for future work would be to incorporate  heterogeneity or local improvements in the data representations.

\subsection{Assumption on model-provider action space}

We make the simplifying assumption that the only action taken by model-providers is to choose a classifier from a pre-specified class. This formalization does not capture other actions (such as data collection and price setting) that may be taken by the platform. Incorporating other model-provider decisions would be an interesting avenue for future work.

\section{Discussion} 

We showed that the monotonicity of scaling trends can be violated under competition. In particular, we demonstrated that when multiple model-providers compete for users, improving data representation quality (as measured by Bayes risk) can \textit{increase} the overall loss at equilibrium. We exhibited the non-monotonicity of the equilibrium social loss in the Bayes risk when representations are varied along several axes (per-representation Bayes risk, noise, dimension, and pre-trained model used to generate the representations). 

An interesting direction for future work is to further characterize the regimes when the equilibrium social loss is monotonic versus non-monotonic in data representation quality as measured by Bayes risk. For example, an interesting open question is to generalize our theoretical results from Section \ref{sec:theory} to more general function classes and distributions of market reputations. Moreover, another interesting direction would be generalize our empirical findings from Section \ref{sec:linear} to other axes of varying data representations and to non-linear classes of predictors. Finally, while we have focused on classification tasks, it would be interesting to generalize our findings to regression tasks with continuous outputs or to generative AI tasks with text-based or image-based outputs. 

More broadly, the non-monotonicity of equilibrium social welfare in scale under competition establishes a disconnect between scaling trends in the single model-provider setting and in the competitive setting. In particular, typical scaling trends (e.g. \citep{K20, SK20, B21, H22, hernandez2021scaling})---which show increasing scale reliably increases predictive accuracy for a single model-provider in isolation---may not translate to competitive settings such as digital marketplaces. Thus, understanding the downstream impact of scale on user welfare in digital marketplaces will likely require understanding how scaling trends behave under competition. We hope that our work serves as a starting point for analyzing and eventually characterizing the scaling trends of learning systems in competitive settings.

\section{Acknowledgments}
We thank Yiding Feng, Xinyan Hu, and Alex Wei for useful comments on the paper. This work was partially supported by the Open Phil AI Fellowship, the Berkeley Fellowship, the European Research Council (ERC) Synergy Grant program, and the National Science Foundation under grant numbers 2031899 and 1804794. 

\printbibliography

\newpage 

\appendix

\section{Additional Details of Simulations }\label{appendix:setup}

\paragraph{Hyperparameters. } We introduce a temperature parameter $\tau$ within our loss function, defining the loss $\ell(f_{w,b}(x), y)$ to be $|\sigmoid((\langle w, x \rangle + b)/\tau) - 1|$. This reparameterizes, but does not change, the model family. 

When we run the best-response dynamics, we always initialize the model parameters as mean-zero Gaussians with standard deviation $\sigma$. When we reinitialize model parameters, we again initialize them as mean-zero Gaussians with standard deviation $\sigma$. For Section \ref{subsec:synthetic}, we set  $I = 5000$,  $\tau = 0.1$, $\epsilon = 0.001$, $\eta = 0.1$, $\sigma = 0.1$, and $\rho = \rho' = 1.0$. For Section \ref{subsec:image} and Section \ref{subsec:image10}, we set $I = 2000$, $\sigma = 0.5$, $\tau = 1.0$, $\epsilon = 0.001$, and $\eta$ with the following learning rate schedule to expedite convergence: $\eta = 1.0$ if the risk $\mathbb{E}_{(x,y) \sim \mathcal{D}}[\ell(f_{w_j, b_j}(x), y)]$ is at least $0.5$, $\eta = 5.0$ if the risk is in $[0.4, 0.5)$, $\eta = 15$ if the risk is in $[0.3, 0.4)$, and $\eta = 20$ if the risk is less than $0.3$. We set $\rho = \rho' = 0.3$ for Section \ref{subsec:image} and we set $\rho = 0.7$ and $\rho' = 1$ for Section \ref{subsec:image10}. 

For Section \ref{subsec:image} and Section \ref{subsec:image10}, we ran over several trials for each data point and the error bars show two standard errors from the mean. For binary classification, the number of trials was 20 for $m = 3$ and $m = 4$ and 8 for $m =5$, $m = 6$, and $m = 8$. For 10-class classification, the number of trials was 40 for $m = 3$ and $m = 4$ and 8 for $m =5$, $m = 6$, and $m = 8$. 

In addition to computing the equilibria, we also approximate the optimal Bayes risk. For Section \ref{subsec:synthetic}, we run gradient descent for $10,000$ iterations with learning rate equal to one and parameters initialized to independent Gaussians with zero mean and standard deviation $0.1$. For Section \ref{subsec:image}, we run gradient descent for $50,000$ iterations with learning rate equal to $0.1$ and parameters initialized to independent Gaussians with zero mean and standard deviation $0.005$. For Section \ref{subsec:image10}, we run gradient descent for $70,000$ iterations with learning rate equal to $0.1$ and parameters initialized to independent Gaussians with zero mean and standard deviation $0.005$.

\paragraph{Generation of the synthetic dataset.}
In Setting 1 (Figures \ref{fig:theory-bayes}, \ref{fig:theory-bayes-varyw}, and \ref{fig:simulations-bayes}), we consider a zero-dimensional population where $Y \mid X$ is distributed as a Bernoulli with probability $\alpha$. In Figure \ref{fig:simulations-bayes}, the meaning of a zero-dimensional representation is that the only parameter is the bias. 

In Setting 2 (Figures \ref{fig:theory-noise}, \ref{fig:theory-noise-varyw}, and \ref{fig:simulations-noise}), we consider a one-dimensional population given by a mixture of Gaussians. In particular, the Gaussian $X \mid Y = 0$ is distributed as $N(-\mu, \sigma^2)$ and the Gaussian $X \mid Y = 1$ is distributed as $N(\mu, \sigma^2)$. The mean $\mu$ is taken to be $1$. The distribution of the labels is given by $\Prob[Y = 1] = 0.4$ and $\Prob[Y = 1] = 0.6$. 

In Setting 3 (Figures \ref{fig:theory-dimension}, \ref{fig:theory-dimension-varyw}, and \ref{fig:simulations-dimension}),  let $D_{\text{base}} = 4$. The distribution over $(X^{\text{all}}, Y)$ consists of $D_{\text{base}}$ subpopulations. We define the distribution of $(X^{\text{all}}, Y)$ as follows: each subpopulation $1 \le i \le D_{\text{base}}$ has a different mean vector $\mu_i \in \mathbb{R}^{D_{\text{base}}}$ and is distributed as $X^{\text{all}} \sim Y = 0 \sim N(-\mu_i, \sigma^2)$, let $X^{\text{all}} \sim Y = 1 \sim N(\mu_i, \sigma^2)$, and let $\Prob[Y = 0] = \Prob[Y = 1] = 1/2$.We define $(\mu_i)_d = 0$ for $1 \le d \le i-1$ and $(\mu_i)_d = 1$ for $i \le d \le D_{\text{base}}$, and we let $\sigma = 1$. If the representation dimension is $D$, then we define $X$ to consist of the first $D$ coordinates of $X^{\text{all}}$. When $D = 0$, the model-provider is not given representations and thus must assign all users to the same output. (Our setup captures that the dimension $D$ must be at least $i$ to see any nontrivial features about subpopulation $i$.) The distribution across the 4 subpopulations is $0.7$, $0.15$, $0.1$, and $0.05$. 

In each case, we draw 10,000 samples and take the resulting empirical distribution to be $\mathcal{D}$. 

\paragraph{Generation of the CIFAR-10 task.}
We consider a binary classification task consisting of the first 10,000 images in the training set of CIFAR-10. The class $0$ is defined to be  \{airplane, bird, automobile, ship, horse, truck\} and class $1$ is defined to be \{cat, deer, dog, frog\}. To generate representations, we use the pretrained models from the Pytorch \texttt{torchvision.models} package; these models were pretrained on ImageNet.

\paragraph{Compute details.} 
We run our simulations on a single A100 GPU.

\section{Additional Results and Proofs for Section \ref{sec:theory}}\label{appendix:theory}

In Appendix \ref{subsec:decomposition}, we show a decomposition lemma and prove existence of equilibrium (Proposition \ref{prop:existence}). We prove the results from Section \ref{subsec:idealized} in Appendix \ref{appendix:proofsidealized}, prove the results from Section \ref{subsec:idealizedmulticlass} in Appendix \ref{appendix:multiclass}, and prove the results from Section \ref{subsec:unequal} in Appendix \ref{appendix:proofsunequal}.

\subsection{Decomposition lemma and existence of equilibrium}\label{subsec:decomposition}

We first show that we can decompose model-provider actions into independent decisions about each representation $x$. To formalize this, let $\mathcal{D}$ be the data distribution, and let $\mathcal{D}_x$ be the conditional distribution over $(X,Y) \mid X = x$ where $(X,Y) \sim \mathcal{D}$. Let $(\mathcal{F}^{\text{multi-class}}_{\text{all}})^x := \left\{f^0, f^1, \ldots, f^{K-1}\right\}$ be the class of $K$ functions from a single representation $x$ to $\left\{0,1, \ldots, K-1 \right\}$, where $f^i(x) = i$. 
\begin{lemma}
\label{lemma:decomposition}
Let $X$ be a finite set of representations, let $\mathcal{F} = \mathcal{F}^{\text{multi-class}}_{\text{all}}$, and let $\mathcal{D}$ be the distribution over $(X,Y)$. For each $x \in X$, let $\mathcal{D}_x$ be the conditional distribution over $(X,Y) \mid X = x$ where $(X,Y) \sim \mathcal{D}$, and let $(\mathcal{F}^{\text{multi-class}}_{\text{all}})^x := \left\{f^0, f^1, \ldots, f^{K-1} \right\}$ be the class of the $K$ functions from a single representation $x$ to $\left\{0,1\right\}$, where $f^i(x) = i$. Suppose that user decisions are noiseless (i.e., $c \rightarrow 0$, so user decisions are given by \eqref{eq:userchoicespecific}). A market outcome $f_1, \ldots, f_m$ is a pure-strategy equilibrium if and only if for every $x \in X$, the market outcome $(f^{f_1(x)}, \ldots, f^{f_m(x)})$ is a pure-strategy equilibrium for $(\mathcal{F}^{\text{multi-class}}_{\text{all}})^x$ with data distribution $\mathcal{D}_x$.
\end{lemma}
The intuition is that since $\mathcal{F}^{\text{multi-class}}_{\text{all}}$ is all possible functions, model-providers make independent decisions for each data representation. 
\begin{proof}
Let $\mathcal{D}^R$ be the marginal distribution of $X$ with respect to the distribution $(X,Y) \sim \mathcal{D}$. First, we write model-provider $j$'s utility as:
\begin{equation}
\label{eq:utilitydecomposed}
  u(f_j; \mathbf{f}_{-j}) = \E_{(x,y) \sim \mathcal{D}} \left[\Prob[j^*(x,y) = j] \right] =  \E_{x' \sim \mathcal{D}^R} \left[\E_{(x,y) \sim \mathcal{D}_{x'}} \left[\Prob[j^*(x,y) = j] \right]\right],
\end{equation}
where $\textbf{f}_{-j}$ denotes the predictors chosen by the other model-providers. The key intuition for the proof will be that the predictions $[f_1(x''), \ldots, f_m(x'')]$ affect $\E_{(x,y) \sim \mathcal{D}_{x'}} \left[\Prob[j^*(x,y) = j] \right]$ if and only if $x' = x''$. 

First we show that if $f_1, \ldots, f_m$ is a pure-strategy equilibrium, then $(f^{f_1(x')}, \ldots, f^{f_m(x')})$ is a pure-strategy equilibrium for $(\mathcal{F}^{\text{multi-class}}_{\text{all}})^{x'}$ with data distribution $\mathcal{D}_{x'}$. Assume for sake of contradiction that $(f^{f_1(x')}, \ldots, f^{f_m(x')})$ is not an equilibrium. Then there exists $j' \in [m]$ such that model-provider $j'$ would achieve higher utility if they switched from $f^{f_{j'}(x')}$ to $f^{l}$ for some $l \neq f_{j'}(x')$. Let $f'_{j'}$ be the predictor given by $f'_{j'}(x) = f_{j'}(x)$ if $x \neq x'$ and $f'_{j'}(x') = l$. By equation \eqref{eq:utilitydecomposed}, this would mean that $u(f'_{j'}; \mathbf{f}_{-j'})$ is strictly higher than $u(f_{j'}; \mathbf{f}_{-j'})$ which is a contradiction.

Next, we show that if $(f^{f_1(x')}, \ldots, f^{f_m(x')})$ is a pure-strategy equilibrium for $(\mathcal{F}^{\text{binary}}_{\text{all}})^{x'}$ with data distribution $\mathcal{D}_{x'}$ for all $x' \in X$ then $f_1, \ldots, f_m$ is a pure-strategy equilibrium. Assume for sake of contradiction that there exists $j'$ such that $u(f'_{j'}; \mathbf{f}_{-j'}) > u(f_j; \mathbf{f}_{-j'})$. By equation \eqref{eq:utilitydecomposed}, there must exist $x'$ such that  $\E_{(x,y) \sim \mathcal{D}_{x'}} \left[\Prob[j^*(x,y) = j'] \right]$ is higher for $f'_{j'}$ than $f_{j'}$. This means that $(f^{f_1(x')}, \ldots, f^{f_m(x')})$ is not an equilibrium (since $f^{l}$ would be a better response for model-provider $j'$) which is a contradiction.  
\end{proof}

We next prove Proposition \ref{prop:existence}, showing that a pure-strategy equilibrium exists by applying the proof technique of Lemma 3.7 of \citet{BT19}. 
\begin{proof}[Proof of Proposition \ref{prop:existence}]
By Lemma \ref{lemma:decomposition}, it suffices to show that there exists a pure-strategy equilibrium whenever there is a single data representation $X = \left\{x\right\}$. In this case, the function class $\mathcal{F}^{\text{multi-class}}_{\text{all}}$ consists of $K$ predictors $\left\{f^0, f^1 \ldots, f^{K-1}\right\}$ given by $f^i(x) = i$. For each class $i$, let $\mathbb{P}[Y = i \mid X] = p_i$. 

For the special case of $K = 2$ (binary classification), the game between model-providers is thus a 2-action game with symmetric utility functions. Thus, it must possess a (possibly asymmetric) pure Nash equilibrium \citep{C04}. 

For the general case of $K \ge 2$, we can no longer apply the result in \citep{C04} since there can be more than 2 actions. We instead show that the game is a potential game, following a similar argument to \citet{BT19}. We define the potential function $\Phi(\cdot)$ as follows. For each $i \in \left\{0, 1, \ldots, K-1\right\}$, we define the function $G_i: \left\{f^0, f^1 \ldots, f^{K-1}\right\}^m \rightarrow \mathbb{R}_{\ge 0}$ to be:
\[
G_i(f_1, \ldots, f_m) := 
\begin{cases}
 \frac{1}{m}  &\text{ if } |\left\{j \in [m] \mid f_j = f^i\right\}| = 0 \\
  \sum_{l=1}^{|\left\{j \in [m] \mid f_j = f^i\right\}| } \frac{1}{l}  &\text{ if } |\left\{j \in [m] \mid f_j = f^i\right\}| \ge 1.
\end{cases}
\]
 We let 
\[\Phi(f_1, \ldots, f_m) := \sum_{i=1}^K p_i \cdot G_i(f_1, \ldots, f_m).\]

We show that $\Phi$ is a potential function for this game. Suppose that model-provider $j$ switches from $f_j := f^i$ to $f'_j = f^{i'}$ for $i' \neq i$. For each $i \in \left\{0, 1, \ldots, K-1\right\}$, let $N_i = |\left\{j \in [m] \mid f_j = f^i\right\}|$ be the number of model-providers who choose $f^i$ on the original outcome $[f_1, \ldots, f_m]$. We observe that:
\[ u(f_j; \mathbf{f}_{-j}) - u(f'_{j}; \mathbf{f}_{-j}) = 
\begin{cases}
p_i \cdot \frac{1}{N_i} - p_{i'} \cdot \frac{1}{N_{i'}+1}  \text{ if } N_i > 1, N_{i'} > 0 \\
p_i \cdot \left(1 - \frac{1}{m} \right) - p_{i'} \cdot \frac{1}{N_{i'}+1} \text{ if } N_i = 1, N_{i'} > 0 \\
p_i \cdot \frac{1}{N_i} - p_{i'} \cdot \left(1 - \frac{1}{m} \right) \text{ if } N_i > 1, N_{i'} = 0 \\
p_i \cdot \left(1 - \frac{1}{m} \right) - p_{i'} \cdot \left(1 - \frac{1}{m} \right)  \text{ if } N_i = 1, N_{i'} = 0 \\
\end{cases}
\]
Moreover,  we see that: 
\begin{align*}
  &\Phi(f_1, \ldots, f_m) - \Phi(f_1, f_2, \ldots, f_{j-1}, f'_j, f_{j+1}, \ldots, f_m) \\
  &=   \sum_{i''=1}^K p_{i''} \cdot G_{i''}(f_1, \ldots, f_m) - \sum_{i''=1}^K p_{i''} \cdot  G_{i''}(f_1, f_2, \ldots, f_{j-1}, f'_j, f_{j+1}, \ldots, f_m) \\
  &= p_i \cdot \left(G_i(f_1, \ldots, f_m) - G_i(f_1, f_2, \ldots, f_{j-1}, f'_j, f_{j+1}, \ldots, f_m)\right)\\
  &+ p_{i'} \left(G_{i'}(f_1, \ldots, f_m) - G_{i'}(f_1, f_2, \ldots, f_{j-1}, f'_j, f_{j+1}, \ldots, f_m)\right).
\end{align*}
If $N_i > 1$, then:
\[G_i(f_1, \ldots, f_m) - G_i(f_1, f_2, \ldots, f_{j-1}, f'_j, f_{j+1}, \ldots, f_m) = \frac{1}{N^i}\]
and if $N_i = 1$, then 
\[G_i(f_1, \ldots, f_m) - G_i(f_1, f_2, \ldots, f_{j-1}, f'_j, f_{j+1}, \ldots, f_m) = 1 - \frac{1}{m}.\]
Similarly, if $N_{i'} > 0$, then:
\[G_{i'}(f_1, \ldots, f_m) - G_{i'}(f_1, f_2, \ldots, f_{j-1}, f'_j, f_{j+1}, \ldots, f_m) = -\frac{1}{N^{i'} + 1}\]
and if $N_{i'} = 0$, then 
\[G_{i'}(f_1, \ldots, f_m) - G_i(f_1, f_2, \ldots, f_{j-1}, f'_j, f_{j+1}, \ldots, f_m) = - \left(1 - \frac{1}{m}\right).\]
Altogether, this implies that:
\[\Phi(f_1, \ldots, f_m) - \Phi(f_1, f_2, \ldots, f_{j-1}, f'_j, f_{j+1}, \ldots, f_m) = u(f_j; \mathbf{f}_{-j}) - u(f'_{j}; \mathbf{f}_{-j}),\]
which shows that $\Phi$ is a potential function of the game. Since pure strategy equilibria exist in potential games \citep{R73, MS96}, a pure strategy equilibrium must exist in the game.   
\end{proof}

\subsection{Proofs for Section \ref{subsec:idealized}}\label{appendix:proofsidealized}

We next prove Proposition \ref{prop:idealized}. The high-level intuition of the proof is as follows. By Lemma \ref{lemma:decomposition}, we can focus on one data representation $x$ at a time. Let $y^* =  \argmax_{y} \Prob[y \mid x]$ be the Bayes optimal label of $x$. The proof boils down to characterizing when the market outcome, $f_j(x) = y^*$ for $j \in [m]$, is an equilibrium, and the equilibrium social loss is determined by whether this market outcome is an equilibrium or not. 
\begin{proof}[Proof of Proposition \ref{prop:idealized}]
Let $\mathcal{D}^R$ be the marginal distribution of $X$ with respect to the distribution $(X,Y) \sim \mathcal{D}$. Let $f^*_1, \ldots, f^*_m$ be a pure-strategy equilibrium. The social loss is equal to:
\begin{align*}
 \texttt{SL}(f^*_1, \ldots, f^*_m) &= \E[\ell(f^*_{j^*(x,y)}(x), y)] \\
 &= \E_{x' \sim \mathcal{D}^R} \left[\E_{(x,y) \sim \mathcal{D}}[\ell(f^*_{j^*(x,y)}(x), y) \mid x = x']\right] \\
 &= \E_{x' \sim \mathcal{D}^R} \left[ \E_{(x,y) \sim \mathcal{D}_{x'}}[\ell(f^*_{j^*(x,y)}(x), y)]\right], 
\end{align*}
where $\mathcal{D}_{x'}$ denotes the conditional distribution $(X,Y) \mid X = x'$ where $(X,Y) \sim \mathcal{D}$. Thus, to analyze the overall social loss, we can separately analyze the social loss on each distribution $\mathcal{D}_{x'}$ and then average across distributions. It suffices to show that $\E_{\mathcal{D}_{x'}}[\ell(f^*_{j^*(x,y)}(x), y)] = \alpha(x')$ if $\alpha(x') < 1/m$ and zero if $\alpha(x') > 1/m$. 

To compute the social loss on $\mathcal{D}_{x'}$, we first apply Lemma \ref{lemma:decomposition}. This means that $(f^*_1(x'), \ldots, f^*_m(x'))$ is pure-strategy equilibrium with $\mathcal{D}_{x'}$. 
We characterize the equilibrium structure for $\mathcal{D}_{x'}$ and use this characterization to compute the equilibrium social loss. 

\paragraph{Equilibrium structure for $\mathcal{D}_{x'}$.} For notational convenience, let $y_i := f^*_i(x')$ denote the label chosen by model-provider $i$ and let 
let $y^* =  \argmax_{y} \Prob[y \mid x']$ be the Bayes optimal label for $x'$. We also abuse notation slightly and let $u(y_1; y_{-j})$ be model-provider $1$'s utility if they choose the label $y_1$ for $x'$ and the other model-provider's choose $y_{-j}$. 

We first show that all model-providers choosing $y^*$ is an equilibrium if and only if $\alpha(x') \le 1/m$. Let's fix $y_j = y^*$ for all $j \ge 2$ and look at model-provider $1$'s utility. We see that $u(y^*; y_{-j}) = 1/m$ and $u(1-y^*; y_{-j}) = \alpha(x')$. This means that $y^*$ is a best-response (i.e., $y^* \in \argmax_y u(y; y_{-j})$) if and only if $\alpha(x') \le 1/m$. 

We next show that if $\alpha(x') < 1/m$, then the market outcome $y_i = y^*$ for all $i \in [m]$ is the only pure-strategy equilibrium. Let $y_1, \ldots, y_m$ be a pure-strategy equilibrium. It suffices to show that $y^*$ is the unique best response to $y_{-j}$; that is, that $\left\{y^*\right\} = \argmax_y u(y; y_{-j})$. To show this, let $m'$ denote the size of the set $\left\{2 \le i \le m \mid y_i = y^*\right\}$. First, if $m' = 0$, then we have that
\[u(y^*; y_{-j}) = 1 - \alpha(x') > 1/m = u(1 - y^*; y_{-j}),  \]
where $1- \alpha(x') > 1/m$ follows from the fact that $1 - \alpha(x') \ge 1/2 \ge 1/m$ along with our assumption that $\alpha(x') \neq 1/m$. This demonstrates that $y^*$ is indeed the unique best response. If $m' = m -1$, then we have that:
\[u(y^*; y_{-j}) = 1/m > \alpha(x') = u(1 - y^*; y_{-j}),  \] as desired. Finally, if $1 \le m' \le m-2$, then:
\[ u(y^*; y_{-j}) = \frac{1-\alpha(x')}{m' + 1} \ge \frac{1-\alpha(x')}{m-1} > \frac{1}{m} > \alpha(x') > \frac{\alpha(x')}{m - m'} = u(1-y^*; y_{-j}), \]
as desired. 

Finally, we show that all model-providers choosing $1-y^*$ is never an equilibrium. Let's fix $y_j = 1 - y^*$ and look at model-provider 1's utility. We see that:
\[ u(y^*; y_{-j}) = 1 - \alpha(x') > \frac{ \alpha(x')}{m} =u(1-y^*; y_{-j}), \]
which shows that $y^*$ is the unique best response as desired.

\paragraph{Characterization of equilibrium social loss.} It follows from \eqref{eq:socialloss} that the equilibrium social loss $\E_{(x,y) \sim \mathcal{D}_{x'}}[\ell(f^*_{j^*(x,y)}(x), y)]$ is $\alpha(x')$ if all of the model-providers choose $y_i = y^*$, it is zero if a nonzero number of model-providers choose $y^*$ and a nonzero number of model-providers choose $1-y^*$, and it is $1-\alpha(x')$ if all of the model-providers choose $1- y^*$.

Let's combine this with our equilibrium characterization results. If $\alpha(x') < 1/m$, then the unique equilibrium is at $y_i = y^*$ so the equilibrium social loss is $\alpha(x)$ as desired. If $\alpha(x') > 1/m$, then neither $y_i = y^*$ for all $i \in [m]$ nor $y_i = 1- y^*$ for all $i \in [m]$ is an equilibrium. Since there exists a pure strategy equilibrium by Proposition \ref{prop:existence}, there must be a pure strategy equilibrium where a nonzero number of model-providers choose $y^*$ and a nonzero number of model-providers choose $1-y^*$. The equilibrium social loss is thus zero. 

Note when $\alpha(x') = 1 - 1/m$, there is actually an equilibrium where all of the model-providers choose $y_i = y^*$, $0$ and an equilibrium where a nonzero number of model-providers choose $y^*$ and a nonzero number of model-providers choose $1-y^*$; thus, the equilibrium social loss can be zero or $1/m$.

\end{proof}

\subsection{Proofs for Section \ref{subsec:idealizedmulticlass}}\label{appendix:multiclass}

We prove Proposition \ref{prop:idealizedmulticlass}.
\begin{proof}[Proof of Proposition \ref{prop:idealizedmulticlass}]

Let $\mathcal{D}^R$ be the marginal distribution of $X$ with respect to the distribution $(X,Y) \sim \mathcal{D}$. Let $f^*_1, \ldots, f^*_m$ be a pure-strategy equilibrium. The social loss is equal to:
\begin{align*}
 \texttt{SL}(f_1, \ldots, f_m) &= \E[\ell(f^*_{j^*(x,y)}(x), y)] \\
 &= \E_{x' \sim \mathcal{D}^R} \left[\E_{(x,y) \sim \mathcal{D}}[\ell(f^*_{j^*(x,y)}(x), y) \mid x = x']\right] \\
 &= \E_{x' \sim \mathcal{D}^R} \left[ \E_{(x,y) \sim \mathcal{D}_{x'}}[\ell(f^*_{j^*(x,y)}(x), y)]\right], 
\end{align*}
where $\mathcal{D}_{x'}$ denotes the conditional distribution $(X,Y) \mid X = x'$ where $(X,Y) \sim \mathcal{D}$. Thus, to analyze the overall social loss, we can separately analyze the social loss on each distribution $\mathcal{D}_{x'}$ and then average across distributions. It suffices to show that 
\[\E_{\mathcal{D}_{x'}} \left[\sum_{i=1}^K \alpha^i(x) \cdot \mathbbm{1}\left[\alpha^i(x) < \frac{c}{m}\right] \right] \le \E_{\mathcal{D}_{x'}}[\ell(f^*_{j^*(x,y)}(x), y)] \le  \E_{\mathcal{D}_{x'}} \left[\sum_{i=1}^K \alpha^i(x) \cdot \mathbbm{1}\left[\alpha^i(x) \le \frac{1}{m} \right] \right]. \]

To compute the social loss on $\mathcal{D}_{x'}$, we first apply Lemma \ref{lemma:decomposition}. This means that $(f^*_1(x'), \ldots, f^*_m(x'))$ is pure-strategy equilibrium with $\mathcal{D}_{x'}$. We then prove properties of the equilibrium structure for $\mathcal{D}_{x'}$ and use these properties to bound the equilibrium social loss. For notational convenience, let $y_i := f^*_i(x')$ denote the label chosen by model-provider $i$ and let 
let $y^* =  \argmax_{y} \Prob[y \mid x']$ be the Bayes optimal label for $x'$. We also abuse notation slightly and let $u(y_1; y_{-j})$ be model-provider $1$'s utility if they choose the label $y_1$ for $x'$ and the other model-provider's choose $y_{-j}$. We can rewrite:
\[\E_{\mathcal{D}_{x'}}[\ell(f^*_{j^*(x,y)}(x), y)] = \E_{\mathcal{D}_{x'}}\left[\sum_{i=1}^K \alpha^i(x) \cdot \mathbbm{1}\left[y_j \neq i \text{ for all } j \in [m] \right] \right]. \]
We first prove the lower bound on $\E_{\mathcal{D}_{x'}}[\ell(f^*_{j^*(x,y)}(x), y)]$ and then we prove the upper bound on $\E_{\mathcal{D}_{x'}}[\ell(f^*_{j^*(x,y)}(x), y)]$.

\paragraph{Proof of lower bound.} Let $y_1, \ldots, y_m$ be a pure strategy equilibrium. To prove the lower bound, it suffices to show that if $\alpha^i(x) < c/m$, then $y_j \neq i$ for all $j \in [m]$. 

Assume for sake of contradiction that $\alpha^i(x) < c/m$ and $y_j = i$ for some $j \in [m]$. Let $i' = \text{argmax}_{i'' \in \left\{0, 1, \ldots, K-1\right\}} \alpha^{i''}(x)$ be the class with maximal conditional probability. By the definition of $c$, we see that $ \alpha^{i'}(x) \ge c > c/m$ which also implies that $i' \neq i$. We split into two cases---(1) $y_{j'} \neq i'$ for all $j' \in \left\{0, 1, \ldots, K-1\right\}$, and (2)  $y_{j'} = i'$ for some $j' \in \left\{0, 1, \ldots, K-1\right\}$---and derive a contradiction in each case.

Consider the first case where $y_{j'} \neq i'$ for all $j' \in \left\{0, 1, \ldots, K-1\right\}$. Then if model-provider $j$ switched from $y_j$ to $i'$, the difference in their utility would be bounded as: 
\begin{align*}
u(i'; y_{-j}) - u(y_j; y_{-j}) &\ge \alpha^{i'}(x) - \left(\frac{\alpha^{i'}(x)}{m} + \alpha^i(x)\right) \\
&= \alpha^{i'}(x) \left(1 - \frac{1}{m} \right) - \alpha^i(x)\\
&> c \left(1 - \frac{1}{m}\right) - \frac{c}{m} \\
&= c \left(1 - \frac{2}{m}\right) \\
&\ge 0,  
\end{align*}
so $y_j$ is not a best-response for model-provider $j$, which is a contradiction. 

Now, consider the second case, where $y_{j'} = i'$ for some $j' \in \left\{0, 1, \ldots, K-1\right\}$. If we compare the utility when model-provider $j$ chooses $i'$ versus $y_j$ as their action, the difference is utility can be bounded as: 
\[u(i'; y_{-j}) - u(y_j; y_{-j}) \ge \frac{\alpha^{i'}(x)}{m} - \alpha^{i}(x) > \frac{c}{m} - \frac{c}{m} = 0. \]
so $y_j$ is not a best-response for model-provider $j$, which is a contradiction. 

This proves the lower bound as desired. 

\paragraph{Proof of upper bound.} Let $y_1, \ldots, y_m$ be a pure strategy equilibrium. To prove the upper bound, it suffices to show if $\alpha^i(x) > 1/m$, then $y_j = i$ for some $j \in [m]$. Assume for sake of contradiction that $\alpha^i(x) > 1/m$ and $y_j \neq i$ for all $j \in [m]$. For any set of actions $y_1, \ldots, y_m$, the total utility $\sum_{j=1}^m u(y_j; y_{-j}) = 1$ sums to $1$. Thus, some model provider $j \in [m]$ must have utility satisfying $u(y_j; y_{-j}) \le 1/m$. However, if model-provider $j$ instead chose action $i$, then they would achieve utility:
\[u(i; y_{-j}) \ge \alpha^i(x) > \frac{1}{m} \ge u(y_j; y_{-j}), \]
so $y_j$ is not a best-response for model-provider $j$, which is a contradiction. This proves the upper bound as desired.

\end{proof}

\subsection{Proofs for Section \ref{subsec:unequal}}\label{appendix:proofsunequal}

A useful lemma is the following calculation of the game matrix when there is a single representation $X = \left\{x\right\}$. 
\begin{lemma}
\label{lemma:gamematrix}
Let $X = \left\{x \right\}$, and let $\mathcal{F} = \mathcal{F}^{\text{binary}}_{\text{all}}$. Suppose that there are $m=2$ model-providers with market reputations $w_{\text{min}}$ and $w_{\text{max}}$, where $w_{\text{max}} \ge w_{\text{min}}$ and $w_{\text{max}} +  w_{\text{min}} = 1$. Suppose that user decisions are noiseless (i.e., $c \rightarrow 0$, so user decisions are given by \eqref{eq:userdecisionsmr}). Then the game matrix is specified by Table \ref{tab:gamematrix}.  
\end{lemma}
\begin{proof}
This follows from applying \eqref{eq:userdecisionsmr} and using the fact that $\ell(y, y') = \mathbbm{1}[y \neq y']$. 
\end{proof}

\begin{table}[]
    \centering
    \begin{tabular}{c|c | c}
     & $y_2 = 1 - y^*$ & $y_2 = y^*$ \\
     \hline
     $y_1 = 1 - y^*$    &  ($w_{\text{max}}, w_{\text{min}}$)   &  ($\alpha(x), 1 - \alpha(x)$) \\ \\
     $y_1 = y^*$    &  ($1 - \alpha(x), \alpha(x)$)   & ($w_{\text{max}}, w_{\text{min}}$)\\
     \\
    \end{tabular}
    \caption{Let $X = \left\{x \right\}$, $\mathcal{F} = \mathcal{F}^{\text{binary}}_{\text{all}}$, user decisions are noiseless, and user decisions are noiseless (i.e., $c \rightarrow 0$, so user decisions are given by \eqref{eq:userdecisionsmr}). Suppose that there are $m=2$ model-providers with market reputations $w_{\text{min}}$ and $w_{\text{max}}$, where $w_{\text{max}} \ge w_{\text{min}}$ and $w_{\text{max}} +  w_{\text{min}} = 1$. Let $y^* =  \argmax_{y} \Prob[y \mid x]$ be the Bayes optimal label for $x'$. The table shows the game matrix when model-provider 1 chooses the label $y_1$ and model provider 2 chooses the label $y_2$.  }
    \label{tab:gamematrix}
\end{table}

We show that pure strategy equilibria are no longer guaranteed to exist when model-providers have unequal market reputations, even when there is a single representation $X = \left\{x \right\}$. 
\begin{lemma}
\label{lemma:nonexistence}
Let $X = \left\{x \right\}$ let $\mathcal{F} = \mathcal{F}^{\text{binary}}_{\text{all}}$. Suppose that there are $m=2$ model-providers with market reputations $w_{\text{min}}$ and $w_{\text{max}}$, where $w_{\text{max}} \ge w_{\text{min}}$ and $w_{\text{max}} +  w_{\text{min}} = 1$. Suppose that user decisions are noiseless (i.e., $c \rightarrow 0$, so user decisions are given by \eqref{eq:userdecisionsmr}). If $\alpha(x) > w_{\text{min}}$, then a pure strategy equilibrium does not exist. 
\end{lemma}
\begin{proof}
For notational convenience, let $y_i := f_i(x')$ denote the label chosen by model-provider $i$ and 
let $y^* =  \argmax_{y} \Prob[y \mid x']$ be the Bayes optimal label for $x'$. We also abuse notation slightly and let $u_i(y; y')$ be model-provider $i$'s utility if they choose the label $y$ for $x$ and the other model-providers choose $y'$. The proof follows from the game matrix show in Table \ref{tab:gamematrix} (Lemma \ref{lemma:gamematrix}). Using the fact that model-provider $1$ must best-respond to model-provider $2$'s action, this leaves $y_1 = 1-y^*$, $y_2 = 1-y^*$ and $y_1 = y^*$, $y_2 = y^*$. However, neither of these market outcomes captures a best-response for model-provider $2$: if $y_1 = 1 - y^*$, then model-provider $2$'s unique best response is $y^*$; if $y_1 = y^*$, then model-provider $2$'s unique best response is $1-y^*$. This rules out the existence of a symmetric or asymmetric pure strategy equilibrium. 
\end{proof}

Given the lack of existence of pure strategy equilibria, we must turn to mixed strategies. A mixed strategy equilibrium is guaranteed to exist since the game has finitely many actions $\mathcal{F}^{\text{binary}}_{\text{all}}$ and finitely many players $m$. Let $(\mu_1, \mu_2, \ldots, \mu_m)$ denote a mixed strategy profile over $\mathcal{F}^{\text{binary}}_{\text{all}}$. We show the following analogue of Lemma \ref{lemma:decomposition} that allows us to  again decompose model-provider actions into independent decisions about each representation $x$. To formalize this, let $\mathcal{D}$ be the data distribution, and again let $\mathcal{D}_x$ be the conditional distribution of $(X,Y)$ when $X = x$, where $(X,Y) \sim \mathcal{D}$. Again, let  $(\mathcal{F}^{\text{binary}}_{\text{all}})^x := \left\{f_0, f_1\right\}$ be the class of the (two) functions from a single representation $x$ to $\left\{0,1\right\}$, where $f_0(x) = 0$ and $f_1(x) = 1$.  Given a mixed strategy profile $\mu$ and a representation $x$, we define the conditional mixed strategy $\mu^x$ over $(\mathcal{F}^{\text{binary}}_{\text{all}})^x := \left\{f_0, f_1\right\}$ to be defined so $\Prob_{\mu^x}[f_i] := \Prob_{f \sim \mu}[f(x) = i]$ for $i \in \left\{0,1\right\}$. 
\begin{lemma}
\label{lemma:decompositionmixed}
Let $X$ be a finite set of representations, let $\mathcal{F} = (\mathcal{F}^{\text{binary}}_{\text{all}})$, and let $\mathcal{D}$ be the distribution over $(X,Y)$. For each $x \in X$, let $\mathcal{D}_x$ be the conditional distribution of $(X,Y)$ given $X = x$, where $(X,Y) \sim \mathcal{D}$, and let $(\mathcal{F}^{\text{binary}}_{\text{all}})^x := \left\{f_0, f_1\right\}$ be the class of the (two) functions from a single representation $x$ to $\left\{0,1\right\}$, where $f_0(x) = 0$ and $f_1(x) = 1$. Suppose that user decisions are noiseless (i.e., $c \rightarrow 0$, so user decisions are given by \eqref{eq:userchoicespecific}). A strategy profile $(\mu_1, \mu_2, \ldots, \mu_m)$ is an equilibrium if and only if for every $x \in X$, the market outcome $(\mu^x_1, \mu^x_2, \ldots, \mu^x_m)$ (where $\mu^x_1$, $\ldots$, $\mu^x_m$ are the conditional mixed strategies defined above) is an equilibrium for $(\mathcal{F}^{\text{binary}}_{\text{all}})^x$ with data distribution $\mathcal{D}_x$.
\end{lemma}
\begin{proof}
The proof follows similarly to the proof of Lemma \ref{lemma:decompositionmixed}, but some minor generalizations to account for mixed strategy equilibria. Let $\mathcal{D}^R$ be the marginal distribution of $X$ with respect to the distribution $(X,Y) \sim \mathcal{D}$. 
Let $\mathcal{D}^R$ be the marginal distribution of $X$ with respect to the distribution $(X,Y) \sim \mathcal{D}$. First, we write model-provider $j$'s utility as:
\begin{equation}
\label{eq:utilitydecomposedgeneral}
 \E_{\substack{f_j \sim \mu_j \\ \mathbf{f}_{-j} \sim \mathbf{\mu}_{-j}}} \left[ u(f_j; \mathbf{f}_{-j}) \right]= \E_{\substack{f_j \sim \mu_j \\ \mathbf{f}_{-j} \sim \mathbf{\mu}_{-j}}} \left[\E_{(x,y) \sim \mathcal{D}} \left[\Prob[j^*(x,y) = j] \right]\right] =  \E_{x' \sim \mathcal{D}^R} \left[\E_{\substack{f_j \sim \mu_j^{x'} \\ \mathbf{f}_{-j} \sim \mathbf{\mu}^{x'}_{-j}}} \left[\E_{(x,y) \sim \mathcal{D}_{x'}} \left[\Prob[j^*(x,y) = j] \right]\right]\right].
\end{equation}
where $\mathbf{\mu}_{-j}$ denotes the mixed strategies chosen by the other model-providers.

First we show that if $\mu_1, \mu_2, \ldots, \mu_m$ is an equilibrium, then $(\mu^{x'}_1, \ldots, \mu^{x'}_m)$ is an equilibrium for $(\mathcal{F}^{\text{binary}}_{\text{all}})^{x'}$ with data distribution $\mathcal{D}_{x'}$. 
Let $f_j$ be in $\text{supp}(\mu_{j'})$. Assume for sake of contradiction that $(\mu^{x'}_1, \ldots, \mu^{x'}_m)$ is not an equilibrium. Then there exists $j' \in [m]$ such that model-provider $j'$ would achieve higher utility on $f^{1-f_{j'}(x')}$ than $f^{f_{j'}(x')}$. Let $f'_{j'}$ be the predictor given by $f'_{j'}(x) = f_{j'}(x)$ if $x \neq x'$ and $f'_{j'}(x') = 1 - f_{j'}(x')$. By equation \eqref{eq:utilitydecomposedgeneral}, this would mean that $u(f'_{j'}; \mathbf{\mu}_{-j'})$ is strictly higher than $u(f_{j'}; \mathbf{\mu}_{-j'})$ which is a contradiction.

Next, we show that if $(\mu^{x'}_1, \ldots, \mu^{x'}_m)$  is an equilibrium for $(\mathcal{F}^{\text{binary}}_{\text{all}})^{x'}$ with data distribution $\mathcal{D}_{x'}$ for all $x' \in X$ then $\mu_1, \ldots, \mu_m$ is an equilibrium. Let $f_j$ be in $\text{supp}(\mu_{j'})$. Assume for sake of contradiction that there exists $j'$ such that $u(f'_{j'}; \mathbf{\mu}_{-j'}) > u(f_j; \mathbf{\mu}_{-j'})$. By equation \eqref{eq:utilitydecomposed}, there must exist $x'$ such that  $\E_{\mathbf{f}_{-j'} \sim \mathbf{\mu}^{x'}_{-j'}} \left[\E_{(x,y) \sim \mathcal{D}_{x'}} \left[\Prob[j^*(x,y) = j'] \right]\right]$ is higher for $f'_{j'}$ than $f_{j'}$. This means that $(\mu^{x'}_1, \ldots, \mu^{x'}_m)$ is not an equilibrium, which is a contradiction. 
\end{proof}

We now prove Proposition \ref{prop:2model-providers}.
\begin{proof}[Proof of Proposition \ref{prop:2model-providers}]
Let $\mathcal{D}^R$ be the marginal distribution of $x$ with respect to the distribution $(x,y) \sim \mathcal{D}$. 
Let $\mu_1, \mu_2$ be a mixed strategy equilibrium. The social loss is equal to:
\begin{align*}
 \E_{\substack{f_1 \sim \mu_1 \\ f_2 \sim \mu_2}}[\texttt{SL}(f_1, f_2)] &= \E[\ell(f_{j^*(x,y)}(x), y)] \\
 &= \E_{\substack{f_1 \sim \mu_1 \\ f_2 \sim \mu_2}}\left[\E_{x' \sim \mathcal{D}^R} \left[\E_{(x,y) \sim \mathcal{D}}[\ell(f_{j^*(x,y)}(x), y) \mid x = x']\right]\right] \\
 &= \E_{\substack{f_1 \sim \mu_1 \\ f_2 \sim \mu_2}}\left[\E_{x' \sim \mathcal{D}^R} \left[ \E_{(x,y) \sim \mathcal{D}_{x'}}[\ell(f_{j^*(x,y)}(x), y)]\right]\right] \\
 &= \E_{x' \sim \mathcal{D}^R} \left[ \E_{\substack{f_1 \sim \mu^*_1 \\ f_2 \sim \mu^*_2}} \left[\E_{(x,y) \sim \mathcal{D}_{x'}}[\ell(f_{j^*(x,y)}(x), y)]\right]\right] \\
 &= \E_{x' \sim \mathcal{D}_X} \left[ \E_{\substack{f_1 \sim \mu^{x'}_1 \\ f_2 \sim \mu^{x'}_2}}\left[\E_{(x,y) \sim  \mathcal{D}_{x'}}[\ell(f_{j^*(x,y)}(x), y)]\right]\right] \\
\end{align*}
where $\mathcal{D}_{x'}$ denotes the conditional distribution $(X,Y) \mid X = x'$ where $(X,Y) \sim \mathcal{D}$ and where $\mu^x$ denotes the conditional mixed strategy $(\mathcal{F}^{\text{binary}}_{\text{all}})^x := \left\{f^0, f^1\right\}$ to be defined so $\Prob_{\mu^x}[f^i] := \Prob_{f \sim \mu}[f(x) = i]$ for $i \in \left\{0,1\right\}$
Thus, to analyze the overall social loss, we can separately analyze the social loss on each distribution $\mathcal{D}_{x'}$ and then average across distributions. It suffices to show that:
\[ \E_{\substack{f_1 \sim \mu^{x'}_1 \\ f_2 \sim \mu^{x'}_2}}\left[\E_{(x,y) \sim \mathcal{D}_{x'}}[\ell(f_{j^*(x,y)}(x), y)]\right] = 
\begin{cases}
 \alpha(x') & \text{ if } \alpha(x') < w_{\text{min}} \\
 \frac{2(\alpha(x') - w_{\text{min}}) \cdot (w_{\text{max}} - \alpha(x))}{(1 - 2 \cdot w_{\text{min}})^2} & \text{ if }\alpha(x') > w_{\text{min}}.
\end{cases}
\]

To compute the social loss on $\mathcal{D}_{x'}$, we first apply Lemma \ref{lemma:decompositionmixed}. This means that $(\mu^{x'}_1, \mu^{x'}_2)$ is mixed-strategy equilibrium with $\mathcal{D}_{x'}$. 
We characterize the equilibrium structure for $\mathcal{D}_{x'}$ and use this characterization to compute the equilibrium social loss.

Our main technical ingredient is the game matrix in Table \ref{tab:gamematrix} (Lemma \ref{lemma:gamematrix}). We will slightly abuse notation and view choosing the label $y$ as the strategy of the model-provider. Accordingly, we view a mixed strategy as a distribution over $\left\{0,1\right\}$. For notational convenience, let $y_i := f_i(x')$ denote the label chosen by model-provider $i$ and let $y^* =  \argmax_{y} \Prob[y \mid x']$ be the Bayes optimal label for $x'$. 
We split into two cases: $\alpha(x') < w_{\text{min}}$ and $\alpha(x') > w_{\text{min}}$.

\paragraph{Case 1: $\alpha(x') < w_{\text{min}}$.} We claim that the unique equilibrium is a pure strategy equilibrium where $y_1 = y_2 = y^*$. First, if $\alpha(x) < w_{\text{min}}$, we show that choosing $y^*$ is a strictly dominant strategy for model-provider $1$. This follows from the fact that $1-\alpha(x) > w_{\text{max}}$ and $w_{\text{max}} \ge w_{\text{min}} >  \alpha(x)$. Thus, model-provider 1 must play a pure strategy where they always choose $y_1 = y^*$. When model-provider 1 chooses $y^*$, then the unique best response for model-provider 2 is also to choose $y^*$ since $\alpha(x') < w_{\text{min}}$. This establishes that $y_1 = y_2 = y^*$ is the unique equilibrium. This also implies that the equilibrium social loss satisfies:
\[\E_{\substack{f_1 \sim \mu^{x'}_1 \\ f_2 \sim \mu^{x'}_2}}\left[\E_{(x,y) \sim \mathcal{D}_{x'}}[\ell(f_{j^*(x,y)}(x), y)]\right] = \alpha(x')\]
as desired. 

\paragraph{Case 2: $\alpha(x') > w_{\text{min}}$.} Let $p_{1} = \mathbb{P}_{\mu^{x'}_1}[y_1 = y^*]$ and let $p_{2} = \mathbb{P}_{\mu^{x'}_2}[y_2 = y^*]$. By Lemma \ref{lemma:nonexistence}, a pure strategy equilibrium does not exist. Thus, we consider mixed strategies. 
Since pure strategy equilibria do not exist, at least one of $p_1$ and $p_2$ must be strictly between zero and one. We compute $p_1$ and $p_2$, splitting into two cases: (1) $p_1 > 0$ and (2) $p_2 > 0$.

If $p_1 > 0$, then we know that model-provider $1$ must be indifferent between choosing $y^*$ and $1-y^*$. This means that:
\[p_2 \alpha(x') + (1-p_2) w_{\text{max}}  = (1 - p_2)(1- \alpha(x')) + p_2 w_{\text{max}}.\]
Solving for $p_2$, we obtain:
\[p_2 = \frac{w_{\text{max}} - (1 - \alpha(x'))}{2 w_{\text{max}} - 1} =  \frac{\alpha(x') - w_{\text{min}}}{1 - 2 w_{\text{min}}}  > 0.\]

If $p_2 > 0$, then we know that model-provider $2$ must be indifferent between choosing $y^*$ and $1-y^*$. This means that:
\[p_1 \alpha(x') + (1 - p_1) w_{\text{min}}  = (1 - p_1)(1- \alpha(x')) + p_1 w_{\text{min}}.\]
Solving for $p_1$, we obtain:
\[p_1 = \frac{(1-\alpha(x')) - w_{\text{min}}}{1 - 2 w_{\text{min}}} = \frac{w_{\text{max}} - \alpha(x)}{1 - 2 w_{\text{min}}} > 0.\]

Putting this all together, we see that:
\begin{align*}
p_1 &= \frac{w_{\text{max}} - \alpha(x')}{1 - 2 w_{\text{min}}}  \\
p_2 &= \frac{\alpha(x') - w_{\text{min}}}{1 - 2 w_{\text{min}}},
\end{align*} 
and in fact $p_1 + p_2 = 1$. 

Using this characterization of $p_1$ and $p_2$, we see that the equilibrium social loss is equal to:
\begin{align*}
\E_{\substack{f_1 \sim \mu^{x'}_1 \\ f_2 \sim \mu^{x'}_2}}\left[\E_{(x,y) \sim \mathcal{D}_{x'}}[\ell(f_{j^*(x,y)}(x), y)]\right] &= 
\alpha(x') \mathbb{P}[y_1 = y^*] \mathbb{P}[y_2 = y^*] + (1 - \alpha(x')) \mathbb{P}[y_1 = 1 - y^*] \mathbb{P}[y_2 = 1 - y^*] \\
&= \alpha(x') p_1 p_2 + (1 - \alpha(x)) (1-p_1) (1-p_2)  \\
&= \alpha(x') p_1 p_2 + (1 - \alpha(x)) p_1p_2  \\
&= p_1 p_2 \\
&=  \frac{(\alpha(x') - w_{\text{min}}) \cdot (w_{\text{max}} - \alpha(x))}{(1 - 2 \cdot w_{\text{min}})^2},
\end{align*}
as desired.

\end{proof}

\end{document}